\documentclass[12pt]{article}

\usepackage{fullpage}

\usepackage{palatino,amsmath,amsthm,amsfonts,amssymb,amscd,bm}
\usepackage{enumitem}
\usepackage{accents}
\usepackage{array}
\usepackage{mleftright}
\usepackage{csquotes}
\usepackage{mleftright}

\usepackage{letltxmacro}
\LetLtxMacro{\oldleft}{\left}
\LetLtxMacro{\oldright}{\right}
\renewcommand{\left}{\mleft}
\renewcommand{\right}{\mright}

\usepackage{thmtools, thm-restate}

\usepackage[labelfont=bf,small]{caption}

\usepackage{xcolor}
\definecolor{citec}{HTML}{324FDA} 
\definecolor{linkc}{HTML}{A0111A}
\definecolor{urlc}{HTML}{b55c87}
\usepackage[
    colorlinks,
    citecolor=citec,
    linkcolor=linkc,
    urlcolor=urlc,
    bookmarks = true,
    breaklinks,
]{hyperref}
\usepackage[nameinlink, noabbrev, capitalize]{cleveref}

\usepackage{tikz}
\usetikzlibrary{decorations.pathreplacing}
\usetikzlibrary{arrows}

\tikzset{
  treenode/.style = {align=center, inner sep=1pt, text centered,
    font=\sffamily},
  arn_n/.style = {treenode, rectangle, white, font=\sffamily\bfseries, draw=black,
    fill=black, text width=1.0em},
  my/.style = {treenode, 
    minimum width=0.5em, minimum height=0.5em}
}

\newtheorem{lemma}{Lemma}[section]
\newtheorem{theorem}[lemma]{Theorem}

\newtheorem{definition}[lemma]{Definition}
\newtheorem{claim}[lemma]{Claim}
\newtheorem{observation}[lemma]{Observation}

\newtheorem{remark}{Remark}
\newtheorem{goal}{Goal}
\newtheorem{question}{Question}

\renewenvironment{proof}[1][]{\begin{trivlist}
\item[\hspace{\labelsep}{\bf\noindent Proof#1:\/}] }{\qed\end{trivlist}} \renewcommand{\qed}{\hfill{\rule{2mm}{2mm}}}

\newcommand{\cC}{\mathcal{C}}
\newcommand{\cV}{\mathcal{V}}
\newcommand{\cB}{\mathcal{B}}
\newcommand{\cS}{\mathcal{S}}
\newcommand\R{{\mathbb{R}}}

\newcommand\F{{\mathbb{F}}}

\newcommand{\N}{{\rm N}}

\newcommand{\cD}{\mathcal{D}}

\newcommand{\cL}{\mathcal{L}}

\newcommand{\moment}{{r}}

\DeclareMathOperator{\poly}{poly}
\DeclareMathOperator*{\E}{\mathbb{E}}
\DeclareMathOperator*{\Var}{Var}

\newcommand\Tr{{\mathop\textup{Tr}}}

\newcommand{\set}[1]{\left\{{#1}\right\}}
\newcommand{\B}{\set{0,1}}

\newcommand{\eps}{\epsilon}

\renewcommand{\mod}{{\rm mod}}
\renewcommand{\eps}{\varepsilon}
\newcommand{\ind}{\mathbf{1}}

\newcommand{\weight}{\mathsf{weight}}

\newcommand{\cin}{\mathcal{C}_\mathrm{in}}
\newcommand{\cout}{\mathcal{C}_\mathrm{out}}
\newcommand{\code}{\mathcal{C}}

\newcommand{\tilc}{\tilde{c}}

\newcommand{\norm}[1]{\left\| #1 \right\|}

\newcommand\dean[1]{\textcolor{cyan}{Dean:~#1}}
\newcommand\mkw[1]{\textcolor{orange}{Mary:~#1}}

\newcommand\ignore[1]{}

\newcommand{\binomial}{\ensuremath{\mathsf{Binomial}}}
\newcommand{\poisson}{\ensuremath{\mathsf{Poisson}}}

\newcommand{\bernoulli}{\ensuremath{\mathsf{Bernoulli}}}

\newcommand{\poi}{\ensuremath{\mathsf{Poi}}}
\newcommand{\ber}{\ensuremath{\mathsf{Ber}}}

\newcommand{\uniform}{\ensuremath{\mathsf{U}}}

\usepackage[T1,T5]{fontenc}

\title{When Do Low-Rate Concatenated Codes Approach
The Gilbert--Varshamov Bound?}
\author{Dean Doron\thanks{Ben-Gurion University of the Negev. \texttt{deand@bgu.ac.il}. Supported in part by NSF-BSF grant \#2022644.} \and Jonathan Mosheiff\thanks{Ben-Gurion University of the Negev. \texttt{mosheiff@bgu.ac.il}. Supported by an Alon Fellowship.} \and Mary Wootters\thanks{Stanford University. \texttt{marykw@stanford.edu}.  Partially supported by NSF grants CCF-2231157 and CCF-2133154.}}

\date{}

\begin{document}

 \maketitle

\begin{abstract}
The \emph{Gilbert--Varshamov} (GV) bound is a classical existential result in coding theory.  It implies that a random linear binary code of rate $\eps^2$ has relative distance at least $\frac{1}{2} - O(\eps)$ with high probability.  However, it is a major challenge to construct \emph{explicit} codes with similar parameters.

One hope to derandomize the Gilbert--Varshamov construction is with code concatenation: We begin with a (hopefully explicit) outer code $\cout$ over a large alphabet, and concatenate that with a small binary random linear code $\cin$.  It is known that when we use \emph{independent} small codes for each coordinate, then the result lies on the GV bound with high probability, but this still uses a lot of randomness.  In this paper, we consider the question of whether code concatenation with a \emph{single} random linear inner code $\cin$ can lie on the GV bound; and if so what conditions on $\cout$ are sufficient for this.

We show that first, there \emph{do} exist linear outer codes $\cout$ that are ``good'' for concatenation in this sense (in fact, \emph{most} linear codes codes are good). We also provide two sufficient conditions for $\cout$, so that if $\cout$ satisfies these, $\cout \circ \cin$ will likely lie on the GV bound.  We hope that these conditions may inspire future work towards constructing explicit codes $\cout$.
\end{abstract}
\thispagestyle{empty}

\newpage

\section{Introduction}\label{sec:intro}

An \emph{error correcting code} (or just a \emph{code}) is a subset $\cC \subseteq \Sigma^n$, for some alphabet $\Sigma$.   
We think of a code $\cC$ being used to encode messages in $\Sigma^k$ for $k = \log_{|\Sigma|}|\cC|$.
That is, for any $m \in \Sigma^k$, we can identify $m$ with a codeword $\cC(m) \in \cC$.\footnote{Here and throughout the paper, we will abuse notation and use $\cC$ both as the code itself (a subset of $\Sigma^n$) and also as an encoding map $\cC \colon \Sigma^k \to \Sigma^n$.}  The idea is that encoding $m$ into the codeword $\cC(m)$ will introduce redundancy that can later be used to correct errors.
In this work we focus on \emph{linear} codes $\cC$, which are codes where $\Sigma = \F$ is a finite field and $\cC \subseteq \F^n$ is a linear subspace of $\F^n$.

Two important properties of error correcting codes are the \emph{rate} $R$ and the \emph{relative distance} $\delta$.  For a code $\cC \subseteq \Sigma^n$, the rate is defined as
$R = \frac{\log_{|\Sigma|}|\cC|}{n} = \frac{k}{n},$
and it quantifies how large the code is.  The rate is between $0$ and $1$, and typically we want it to be as close to $1$ as possible; this means that  the encoding map does not introduce much redundancy.  The (relative) distance of $\cC\subseteq \Sigma^n$ is defined as
$\delta = \frac{1}{n} \min_{c \neq c' \in \cC} \Delta(c,c')$,
where $\Delta(\cdot, \cdot)$ is Hamming distance.  Again, the relative distance is between $0$ and $1$, and again we typically want it to be as close to $1$ as possible; this means that the code can correct many worst-case errors.

These two quantities---rate and distance---are in tension.  The larger the rate is, the smaller the distance must be.  For binary codes (that is, codes where $\Sigma = \F_2$), it is a major open question to pin down the best trade-off possible between rate and distance.  However, we know that good trade-offs are possible: 
The best known possibility result in general is the \emph{Gilbert--Varshamov} (GV) bound (\cref{thm:GV}).  

In this paper we focus on \emph{low rate} codes.  In this parameter regime, the GV bound implies that there \emph{exist} binary linear codes
with relative distance $\frac{1-\eps}{2}$ and rate $\Omega(\eps^2)$, for small $\eps > 0$. In fact, Varshamov's proof shows that a random binary linear code achieves this with high probability. 

Constructing such codes explicitly, hopefully 
accompanied by an efficient decoding algorithm, has been subject to extensive and 
fruitful research in the past decades (e.g., \cite{NN90,ABNNR92,AGHP92,BT13,GI05,T17,BD22}), with several exciting 
breakthroughs in recent years. 
These breakthroughs include explicit constructions of codes with distance $\delta = \frac{1-\eps}{2}$ and rate $R = \Omega(\eps^{2 + o(1)})$, even with efficient algorithms (see \cref{sec:related}).
However, there are still open questions.  For example, we do not know how to attain $\delta = \frac{1 - \eps}{2}$ and $R = \Omega(\eps^2)$ (without any $o(1)$ term) explicitly, and we do not have explicit constructions approaching the GV bound with rates bounded away from zero. 
Motivated by these questions, we consider \emph{concatenated codes}, possibly with some randomness, which we discuss next.

\paragraph{Concatenated Codes, and Our Question.} 
A natural candidate for explicit (for low randomness) codes on the GV bound are \emph{concatenated linear codes.}  These codes are built out of two ingredients: a (hopefully explicit) linear outer code $\cout \subseteq \F_q^n$ with dimension $k$ for some large $q$; and a smaller inner binary linear code $\cin \subseteq \F_2^{n_0}$, with dimension $k_0 = \log_{2}q$.  We define the concatenated code $\cC = \cout \circ \cin \subseteq \F_2^{n_0 \cdot n}$ by first encoding a message $m \in \F_q^k$ (which can also be thought of as $m \in \F_2^{k_0 \cdot k}$) with $\cout$.  Then, we encode each symbol of the resulting codeword using $\cin$.  That is, for a message $m$,
\[ \cC(m) =  (\cin(\cout(m)_1), \cin(\cout(m)_2), \cdots, \cin(\cout(m)_n) ) \in \F_2^{n_0 \cdot n}.\]
It is not hard to see that the rate of $\cC$ is the product of the rates of $\cin$ and $\cout$, and that the distance of $\cC$ is \emph{at least} the product of the distances of $\cin$ and $\cout$.

The natural approach to constructing a good concatenated code is to choose $\cout$ and $\cin$ with the best known trade-offs: Since $\cout$ is over a large alphabet, we know explicit constructions of codes with optimal rate-distance trade-off\footnote{For codes over large alphabets, the best possible trade-off is the \emph{Singleton bound}, or $R = 1 - \delta$.  This is achievable, for example, by Reed--Solomon codes.}; and if $n_0$ is sufficiently small, we can find a $\cin$ on the GV bound either deterministically by brute force or else with low randomness, depending on the size of $n_0$.

However, in general this approach will not achieve the GV bound.  
 If we do not assume any additional properties of $\cout$ and $\cin$, and simply use the concatenation properties, 
then setting the parameters so that $\code = \cout \circ \cin$ has distance
$\frac{1-\eps}{2}$, the rate of $\code$ will be at most roughly $\eps^3$.
This is known as the \emph{Zyablov bound} \cite{Zyablov} (see also \cite{GRS}).  
As we discuss more in \cref{sec:related}, concatenation has been a popular approach to obtain fully explicit codes with good rate-distance trade-offs, but none of these constructions are known to beat the Zyablov bound.

Instead of using a \emph{single} inner code, several works have focused on a related construction originally due to  Thommesen \cite{Thom}, which uses multiple inner codes.  More precisely, this construction uses i.i.d.\ random linear inner codes for each coordinate.  
It can be shown~\cite{Thom} that the resulting code does lie on the GV bound with high probability, and if $\cout$ is chosen appropriately there are even efficient decoding algorithms for it~\cite{GI04,atrithesis, HRZW19}. 
However, this approach relies heavily on the fact that the inner codes are independent, and as a result uses a lot of randomness.  

This state of affairs motivates the following question (also asked in the title of this paper):
\begin{question}\label{q:main}
 Are there concatenated linear codes $\cout \circ \cin$ (with a \emph{single} random linear inner code $\cin$) that meet the GV bound with high probability over $\cin$?\footnote{Of course, if the length of either the inner code or the outer code is $1$, this question reduces to the non-concatenated setting; we are interested in parameter regimes where $n_0$ is non-trivial.}  If so, are there sufficient conditions on $\cout$ that will guarantee this?
\end{question}

In this paper, we show that \emph{yes}, there are concatenated codes that meet the GV bound, and we also give two sufficient conditions on $\cout$ for this to hold.  Our existential result is non-constructive, but it is our hope that our sufficient  conditions will lead to explicit constructions of appropriate $\cout$-s, which would lead to explicit (or at least pseudo-random, depending on the alphabet size of $\cout$) concatenated codes on the GV bound. 

\begin{remark}[Motivation for \cref{q:main}]
Above, we have motivated \cref{q:main} as an avenue towards explicit or pseudo-random binary codes on the GV bound, and indeed this is our original motivation.  But we point out that \cref{q:main} is also interesting in its own right.  Concatenated codes are a classical construction, going back to the 1960's~\cite{forney}, and have been used in many different settings over the decades.  It seems like a fundamental question to understand when these codes can attain the GV bound.
\end{remark}

\begin{remark}[Focus on Linear Codes]
    In \cref{q:main} and in this paper, we focus on \emph{linear} codes.  This is because if we used, say, a uniformly random non-linear code as the inner code, it would require exponentially more randomness than a random linear inner code, so this does not seem like a hopeful avenue for derandomization.  We note however that the question is much easier for non-linear codes.  For example, suppose that $\cout$ is a Reed--Solomon code of rate $\eps$ so that each symbol is additionally tagged with its evaluation point: that is, the symbol corresponding to $\alpha \in \F_q$ is $(\alpha, f(\alpha)) \in \F_q^2$.  For the inner code, we use a completely random (non-linear) code of rate $\eps$.  Then since all of the symbols in each outer codeword are different by construction, each codeword is essentially uniformly random, and it is not hard to show that the result is close to the GV bound in the sense that a code of rate $O(\eps^2)$ will have distance $1/2 - O(\eps)$ with high probability.  This same argument will not work when $\cin$ is linear, since the different symbols of codewords of $\cout$ will still have $\F_2$-linear relationships.
\end{remark}

\paragraph{Our Contributions.}
Our main results are:
\begin{enumerate}
    \item \textbf{Existence of concatenated codes on the GV bound.} We answer the first part of \cref{q:main}: there \emph{are} concatenated codes $\cout \circ \cin$ that achieve the GV bound, in a wide variety of parameter regimes.  In particular, we show that \emph{most} codes $\cout$ are actually good:
    \begin{theorem}[Informal;  \cref{thm:main-random}]\label{thm:main-intro}
Suppose that $\cout \subseteq \F_q^n$ and $\cin \subseteq \F_2^{n_0}$ are random linear codes of rate $\eps$, so that $q \geq 2^{\Omega(\eps^{-3})}$.  Then $\cC = \cout \circ \cin$ has rate $\eps^2$, and with high probability, the relative distance of $\cC$ is at least $1/2 - O(\eps)$.
\end{theorem}

While \cref{thm:main-intro} seems intuitive (in the sense that a random linear code lies on the GV bound with high probability, so why not concatenated random linear codes?), to the best of our knowledge it has not appeared in the literature before, and the proof was not obvious (to us).\footnote{We note that earlier work by Barg, Justesen and Thomessen~\cite{BJT01} also addresses random linear outer codes concatenated with an arbitrary (fixed) inner code, using very different techniques than we do.  They do not explicitly state a statement like \cref{thm:main-intro} above, though it is plausible that their techniques could be used to prove something similar.  We discuss their techniques and the relationship to our work in \cref{sec:related}.}  One challenge is that a codeword $c \in \cout \circ \cin$ is not uniformly random in $\F_2^{N}$.  In particular, the natural strategy of ``show that each non-zero codeword has high weight with high probability and union bound'' that is used to establish the Gilbert--Varshamov bound will not work in this setting, as we do not have enough concentration.

    \item \textbf{Sufficient conditions for $\cout$.}
    Our existence result above uses a random linear code as the outer code, which does not help in the quest for explicit constructions.  However, our proof techniques inspire two sufficient conditions on $\cout$.  That is, if $\cout$ satisfies these conditions, then $\cout \circ \cin$ will meet the GV bound with high probability when $\cin$ is a random linear code.  Our hope is that formalizing these will lead to explicit constructions in the future.

    We give an overview and intuition for our two sufficient conditions here.  We note that both conditions are only sufficient when the alphabet size $q$ for $\cout$ is suitably large (exponential in $1/\poly(\eps)$); see \cref{thm:suff} and \cref{main-entropy} for details.
    \begin{itemize}
        \item \textbf{Sufficient Condition 1: A soft-decoding-like condition on $\cout^\perp$.}
        Our first sufficient condition, formalized in \cref{thm:suff}, is a soft-list-decoding-like condition on $\cout^\perp$.  More precisely, we define a distribution $\cD$\footnote{
The distribution $\cD$ is intuitively defined as follows.  Let $\cin$ be the inner code, and suppose that it has a generator matrix $G_0 \in \F_2^{n_0 \times k_0}$. 
 Then to sample from $\cD$, we take a random sparse linear combination of the rows of $G_0$ (over $\F_2$), and interpret the result in $\F_2^{k_0}$ as an element of $\F_q$, which we return.} 
 on the alphabet $\F_q$; the condition is that
 \begin{equation}\label{eq:cond_intro}
        \Pr_{x \sim \cD^n}[ x \in \cout^\perp \setminus \{0\} ] \leq \frac{1}{q^k}(1 + \Delta)
        \end{equation}
        for some small $\Delta$.
        Note that $1/q^k$ is the probability that a completely random vector is in $\cout^\perp$,
        so this condition is saying that if the coordinates of $x$ are drawn i.i.d.\ from the same distribution $\cD$, then $x$ not much more likely to be in $\cC^\perp$ than in a uniformly random vector.
        We show that if this holds, then $\cout \circ \cin$ lies on the GV bound with high probability over the choice of a random linear inner code $\cin$.

 It's not hard to see (\cref{rem:exists}) that this condition holds in expectation for a random linear code $\cout$, and in particular there exist linear codes $\cout$ that have this property.

        This condition is reminiscent of $\cout^\perp$ being list-decodable from soft information (e.g., \cite{KV03}).  In soft-list-decoding, one typically gets a distribution $\cD_i$ for each $i \in [n]$, interpreted as giving ``soft information'' about the $i$'th symbol. If one can show that a vector drawn from $\cD_1 \times \cdots \times \cD_n$ is unlikely to be in the code, this implies that there are not too many codewords that are likely given the soft information we hare received.
        However, there are several differences between existing work on soft list-decoding and our work, notably that our distribution $\cD$ is a particular one and is the same for all $i$, and also there are some differences in the parameter settings.

        This condition can also be seen as a soft form of \emph{list-recovery}, where we have the same list in each coordinate.\footnote{Informally, a code $\cC \subseteq \Sigma^n$ is said to be list-recoverable if for any small sets $S_1, \ldots, S_n \subseteq \Sigma$, there are not too many codewords $c \in \cC$ so that $c_i \in S_i$ for many values of $i$.}  In more detail, if the support of $\cD$ is concentrated on a small set $S$ (which ours is for reasonable settings of $n_0, \eps$, see \cref{rem:listrec}), then the condition in \cref{thm:suff} is related to asking that the number of codewords that lie in the combinatorial rectangle given by $S \times S \times \cdots \times S$ is about what it should be.  Unfortunately, the definition of ``small'' here does not seem to be small enough for existing constructions of list-recoverable codes (for example folded RS codes or multiplicity codes) to yield any results.
       
        \item \textbf{Sufficient Condition 2: $\cout$ has good min-entropy.} Our second sufficient condition, formalized in \cref{main-entropy}, requires the codewords of $\cout$ to be ``smooth'', meaning, roughly, that every nonzero codeword has a fairly uniform distribution of symbols from $\F_q$. To illustrate why a smoothness condition is desirable, let us consider two extreme cases.

        The bad extreme is when there exists a codeword $c$ that is supported on very few symbols, say even on a single symbol. If $c = (\sigma, \sigma, \ldots, \sigma)$ for some $\sigma \in \F_q$, then the relative weight of $c \circ \cin$, for a random binary inner code $\cin$ of rate $\eps$, might be $\frac{1}{2}-\Omega(\sqrt{\eps})$, much worse than the $\frac{1}{2} - O(\eps)$ that we would want for the GV bound.  

          The good (possibly unrealistic) extreme is where each nonzero codeword of $\cout$ has a symbol distribution that is \emph{uniform} over $\F_q$.  In this case it is not hard to see that $\cout \circ \cin$ will be close to the GV bound with high probability over a random linear code $\cin$.  (For this, all we need is that $\cin$ has about the ``right'' weight distribution, which a random linear code will have with high probability).

          The natural question is thus \emph{how smooth} the codewords of $\cout$ should be in order for $\code$ to have distance $\frac{1}{2}-O(\eps)$. In \cref{sec:minent}, we quantify this by the \emph{smooth min-entropy} of the codewords' empirical distributions on symbols.  We show in \cref{main-entropy} that if this smooth min-entropy is large enough for all $c \in \cout$, then $\cC = \cout \circ \cin$ is likely to lie near the GV bound when $\cin$ is a random linear binary code.

        How large is ``large enough''?
          For this informal discussion, we give one example of the parameter settings from \cref{main-entropy}: It is enough for every non-zero codeword $c \in \cout$ to have a symbol distribution that has $\Theta(\eps n)$ copies of the same symbol (say, the zero symbol), while the remaining symbols in $c$ are uniformly distributed over a set of size only $q^{1 - \eps}$.  By some metrics this is still a fairly ``spiky'' distribution, but it is ``smooth enough'' for our purposes.

        Note that while our soft-decoding-like condition considers $\cout^{\perp}$, our smooth min-entropy condition here considers $\cout$ itself.
    \end{itemize}
\end{enumerate}

\subsection{Related Work}\label{sec:related}

\paragraph{Explicit Concatenated Codes.}
Concatenation (with a single inner code) has been a common approach to obtain explicit codes close to the GV bound.  Here we mention a few such places this comes up. 
 Choosing $\cout$ to be the Reed--Solomon code, and $\cin$ to be the Hadamard code, gets a code of length $O(k^2/\eps^2)$ for any dimension $k$ \cite{AGHP92}, and replacing Reed--Solomon with the Hermitian code gets length $O((k/\eps)^{5/4})$ \cite{BT13}. Choosing a  different AG code for $\cout$ can result in non-vanishing rate and in fact approach rate $\eps^3$ (see~\cite{T17}).  Moreover,
concatenating Reed--Solomon with the Wozencraft ensemble gives the \emph{Justesen} code \cite{Jus}, having constant relative rate and constant relative distance. Note that none of these concatenation-based constructions thus far have beat the Zyablov bound.

\paragraph{Concatenated Codes with Random Linear $\cout$.}  
Relevant to \cref{thm:main-intro}, \cite{BJT01} studies a random linear code $\cout$ concatenated with a \emph{fixed} inner code $\cin$.  (See also \cite{BM10}, which applies the same techniques for an application in compressive sensing).  The work \cite{BJT01} derives bounds on the distance of $\cout \circ \cin$ in terms of (moments of) the weight distribution of $\cin$.  These bounds imply that $\cout \circ \cin$ approaches the GV bound in some cases, but doesn't seem to immediately imply \cref{thm:main-intro}. 

Before discussing their techniques more, we note that the biggest difference between \cite{BJT01} and our work is that their question is about the behavior of random linear codes, and so naturally their approach crucially uses the fact that $\cout$ is random.  In contrast, the motivation for our work is to find deterministic sufficient conditions on $\cout$, and we invoke a random linear outer code as a proof of concept that our approach is realizable.

Next, we briefly describe the techniques and implications of \cite{BJT01}, relative to \cref{thm:main-intro}.
The key result of \cite{BJT01} is an expression of the limiting trade-off between the rate $R$ and the distance $\delta$ of $\cout \circ \cin$, in terms of the function $\phi(\tau) = \ln \mathbb{E}_X[e^{\tau X}]$, where $X$ is the weight of a random codeword from $\cin$ and where $\tau \leq 0$ parameterizes the trade-off.\footnote{In more detail, this trade-off is given by $R = \frac{1}{n_0 \ln(2)} (\tau \phi'(\tau) - \phi(\tau))$ and $\delta = \frac{\phi'(\tau)}{n_0},$ for $\tau \leq 0$.}
They show that this trade-off meets the GV bound when $\cin$ is the identity (trivial) code, and investigate how it behaves when $\cin$ is a non-trivial code.
Towards this, one can use their trade-off to work out the Taylor series for $R$ around $\delta = 1/2$.  
It is not hard to see that under mild conditions on $\cin$, the first two terms of this Taylor expansion vanish and hence we obtain $R = \Theta(\eps^2) + O_{\cin}(\eps^3)$ when $\delta = 1/2 - \eps$, where the $O_{\cin}(\cdot)$ notation hides constants that depend on $\cin$.  This implies that if $n_0$ is a constant, independent even of $\eps$, then $\cout \circ \cin$ approaches the GV bound.  However, if $n_0$ is growing relative to $\eps$ (which it is in our case, as we take $\cin$ to have rate $\eps$), then the ``constant'' terms hiding in the $O_{\cin}(\eps^3)$ term may depend on $n_0$, which in turn may depend on $\eps$.  It seems plausible that when $\cin$ is a random linear code, this dependence is mild\footnote{In particular, as pointed out in \cite{BJT01}, the first $d^\perp - 1$ terms of the Taylor series will agree with the GV bound, where $d^\perp$ is the dual distance of $\cin$, which for a random linear code $\cin$ is quite large.} and something like \cref{thm:main-intro} could be established with these techniques, but to the best of our knowledge such a proof has not appeared in the literature and does not seem to follow immediately.

\paragraph{Non-Concatenation-Based Explicit Constructions.}
As mentioned above, there have been several breakthroughs in the past few years obtaining explicit constructions of binary codes near the GV bound, and even efficient algorithms for them.
In a breakthrough result, Ta-Shma \cite{T17} constructed \emph{explicit} linear codes
of relative distance $\frac{1-\eps}{2}$ having rate $\eps^{2+o(1)}$. Ta-Shma's codes
are also \emph{$\eps$-balanced}, i.e., $\Delta(x,y) \in \left[ \frac{1-\eps}{2},\frac{1+\eps}{2} \right]$, and thus give rise to explicit $\eps$-biased sample spaces, which are ubiquitous in pseudorandomness and derandomization. Works that followed gave efficient \emph{decoding} of Ta-Shma codes and their variants \cite{alev2020list,JQST20,JST21,SS23,jeronimo2023list} (see also \cite{BD22} for a different, randomized, construction that slightly improves upon the rate of \cite{T17}, and admits efficient decoding).
We note that these codes are graph-based, and do not in general have a concatenated structure.

\paragraph{Results with Multiple i.i.d.\ Inner Codes.}
Thommesen showed that when the outer code is a Reed--Solomon code, and it is concatenated with $n$ different random linear codes, one for each coordinate, chosen independently, then the resulting code lies on the GV bound with high probability~\cite{Thom}.
Guruswami and Indyk devised efficient decoding algorithms for these codes, based on \emph{list-recoverability} of the outer code~\cite{GI04}. That work used a Reed--Solomon code as the outer code, which is list-recoverable up to the Johnson bound.  Later, Rudra~\cite{atrithesis} observed that the parameters could be improved by swapping out the Reed--Solomon code for a code that can be list-recovered up to capacity, for example a Folded Reed--Solomon code.  Later work obtained nearly-linear-time decoding algorithms by swapping out the outer code for a capacity-achieving list-recoverable code with near-linear-time list-recovery algorithms~\cite{HRZW19,KS20}. Codes with multiple i.i.d.\ inner codes have also been studied in \cite{Zyablov71,BZ73}. 

We also mention the work of Guruswami and Rudra~\cite{GR10}, who show that the same construction (a list-recoverable code concatenated with $n$ different i.i.d. random linear codes) is list-\emph{decodable} up to capacity with high probability.  In the results \cite{GI04,atrithesis,HRZW19,KS20} mentioned above, list-recovery of the outer code was needed for \emph{algorithms}, not the combinatorial result (which follows already from~\cite{Thom}).  In contrast, in \cite{GR10}, the list-recoverability of the outer code is needed for the combinatorial result itself.  In that sense, the flavor is similar to our sufficient condition in \cref{sec:soft}, although the techniques are very different, and in our work we only use one inner code.

\paragraph{Further Low-randomness Constructions of Binary Codes on GV Bound.}
If one's goal is to explicitly construct a binary code that achieves that GV bound, at least two types of partial results may be considered as subgoals. In the first class of results, one seeks explicit codes whose rate vs.\ distance tradeoff is as close to the GV bound as possible. This includes the works discussed in the first two paragraphs of \cref{sec:related} above. 
A second path is to seek codes that fully attain the GV bound, and strive to minimize the amount of randomness used in their construction.

Varshamov's classic result \cite{Varshamov} is that a random linear code likely achieves the GV bound. Constructing such a code of length $n$ and rate $R$ requires sampling either a random generating matrix or a random parity-check matrix, and thus $O\left(\min\{R,1-R\}\cdot n^2\right)$ random bits are needed. Two classical elementary constructions---the Wozencraft ensemble \cite{Massey1963} and the random Toeplitz Matrix construction (e.g., \cite[Exercise 4.6]{GRS}) 
---are able to reduce the needed randomness to $O(n)$. 

So far, no codes achieving the GV bound using $o(n)$ randomness are known. Moreover, there is a certain natural obstacle, which we now describe, that needs to be tackled before sublinear randomness can be achieved. Say that a random code $\cC\subseteq \F_2^n$ is \emph{uniform} if every $x\in \F_2^n\setminus\{0\}$ appears in the code with the same probability, namely, $p_{R,n} = \frac{2^{Rn}-1}{2^n-1}$. It is not hard to prove via a union bound that a uniform linear code achieves the GV bound with high probability (this is exactly Varshamov's observation). To the best of our knowledge, every known GV-bound construction to date, including the linear randomness constructions mentioned above, is uniform. Unfortunately, a uniform code ensemble with sublinear randomness cannot exist as long as $R$ is bounded away from $1$. Indeed, to have events that occur with probability $p_{R,n}$, at least $\log_2 \frac1{p_{R,n}} \approx (1-R)n$ random bits are required. Therefore, a code construction obtaining the GV bound with sublinear randomness would have to do so without being uniform (see also \cite[Section 5]{MRSY24}). 
We have hope that our sufficient conditions in \cref{thm:suff,main-entropy} could be attained by non-uniform codes.  For example, as discussed above, the soft-decoding-like condition of \cref{thm:suff} is reminiscent of results on soft-list-decoding and soft-list-recovery, which in different parameter regimes can even be achieved by {deterministic} codes.  

A related line of work \cite{GM22,PP23,MRSY24} attempts to construct codes that enjoy a broad class of desirable combinatorial properties similar to those of random linear codes using as little randomness as possible. Such properties include not just the GV bound, but also list decodability up to the \emph{Elias bound} (see \cite{MRSY24}), list recoverability, and, more generally, \emph{local similarity} (see \cite[Definition 2.14]{MRSY24}) to a random linear code. 

\subsection{Technical Overview}\label{sec:tech}
In this section we give an overview of the main technical ideas.  This section also serves as an outline of the paper.

\paragraph{\cref{sec:moment}: A moment-based framework.} 
In \cref{sec:moment}, we set up a framework that will be useful for the results in \cref{sec:random} and \cref{sec:soft}.  We describe this approach here.

Suppose that we are trying to encode a message $m \in \F_q^k$ with our concatenated code $\cC = \cout \circ \cin$, to obtain $\cC(m) = w \in \F_2^{n \cdot n_0}$.  
Each symbol of $w$ is indexed by some $\alpha \in [n]$ and some $\beta \in [n_0]$; this symbol is equal to
\[ (\cin(\cout(m)_\alpha))_\beta = \langle \cout(m)_\alpha, b_\beta \rangle,\]
where $b_\beta$ is the $\beta$'th row for a generator matrix
 $G_0 \in \F_2^{n_0 \times k_0}$ for $\cin$, and 
where the $\langle \cdot, \cdot \rangle$ notation denotes the dot product over $\F_2$.  This motivates the definition of a variable $X_m \in \R$ defined by
\[ X_m = \sum_{\alpha \in [n]} \sum_{\beta \in [n_0]} (-1)^{\langle \cout(m)_\alpha, b_\beta \rangle}.\]
Indeed, $X_m$ is the bias of $w = \cC(m)$; the weight of $w$ is at least $\frac{1}{2} - O(\eps N)$ if and only if $X_m$ is at most $O(\eps N)$.
Thus, to show that the code $\cC$ has distance at least $\frac{1}{2} - O(\eps N)$, it suffices to show that 
\[ \max_{m \in \F_q^k \setminus \{0\}} X_m = O(\eps N).\]

Our strategy will be to consider a large moment of $X_m$ over the choice of a random nonzero message $m$:
\[ \mathbb{E}_{m \sim \F_q^k \setminus \{0\}} [X_m^r] \]
for some appropriate $r$.
If we can show that this is smaller than $(c \eps N)^r / q^k$, then Markov's inequality will imply that 
\[ \Pr_{m \sim \F_q^k \setminus \{0\}}[ X_m \geq c \eps N ] \leq \frac{ \mathbb{E}_{m \sim \F_q^k \setminus \{0\}} [X_m^r]}{(c\eps N)^r} < \frac{1}{q^k}, \]
and in particular that there are no messages $m$ so that $X_m \geq c \eps N$.

In \cref{lem:moment}, we take a Fourier transform in order to re-write $\E[X_m^r]$ as a quantity involving $\cout^\perp$.  This quantity can be thought of as follows.  For every integer-valued matrix\footnote{In the actual quantity, the entries of this matrix are ordered, and we denote it $\cV$ instead of $V$; we ignore the ordering in this discussion for simplicity.} $V \in \mathbb{Z}_{\geq 0}^{n_0 \times n}$ with entries that sum to $r$, we consider a vector $g_V \in \F_q^n$ defined by considering the matrix $G_0^T \cdot V \in \F_2^{k_0 \times n}$ and then treating it as a vector $g_V \in \F_q^n$ by identifying each of the columns in $\F_2^{k_0}$ with elements of $\F_q$.  Then the quantity in \cref{lem:moment} has to do with the number of these vectors $g_V$ that are in $\cout^\perp$.  The exact expression doesn't matter too much for this informal discussion; instead we explain below how we use this re-writing to prove \cref{thm:main-random} and \cref{thm:suff}.

\paragraph{\cref{sec:random}: Most codes $\cout$ are good.}
\cref{thm:main-random} informally says that if $\cout$ is a random linear code, then with high probability $\cout \circ \cin$ is near the GV bound.  
In the proof, we use our framework from \cref{sec:moment}, and show that with high probability over $\cout$, the moment $\E_m[X_m^r]$ is small for an appropriate $r$.  To do this, we need to count the number of matrices $V$ described above that are likely to land in $\cout^\perp$.  Since $\cout$ is a random linear code, so is $\cout^\perp$, and so the probability of any particular \emph{non-zero} $g_V$ landing in it is small (about $1/q^k$), while of course the probability that $0$ is contained in $\cout^\perp$ is $1$.  Thus, the challenge is understanding how many $g_V$-s are actually zero.  There are two ways that a matrix $V$ as described above could lead to $g_V = 0$: Either $V = 0 \text{ mod }2$, or else $V$ is non-zero mod $2$ but $G_0^T V = 0$.  The first case can be counted straightforwardly.  For the second, we leverage the weight distribution that the inner code $\cin$ is likely to have.  We note that this is the only place (in any of our arguments) that we need $\cin$ to be a random linear code: We just need it to have approximately the ``right'' weight distribution.

\paragraph{\cref{sec:soft}: A soft-decoding-like sufficient condition.}
The expression that we get for $\E_m[X_r^m]$ in \cref{lem:moment} directly inspires our soft-decoding-like sufficient condition in \cref{thm:suff}.  One can view the task of counting the matrices $V$ so that $g_V \in \cout^\perp$ as choosing a random $V$ and asking about the probability that $g_V \in \cout^\perp$.  If the columns of $V$ were independent, then this would be the same as choosing the coordinates of $g_V$ i.i.d.\ from some distribution $\cD$.  Thus we would get a requirement on
$\Pr_{x \sim \cD^n}[x \in \cout^\perp],$
similar to the condition in \cref{eq:cond_intro} that we end up with.

Of course, the coordinates are not independent (because the total weight of $V$ is fixed to be $r$), but this can be solved. 
 In more detail, we choose $r$ to be a Poisson random variable, which in this setting makes the columns of $V$ independent.  One hiccup is that the ``Poisson-ized'' distribution turns out to be meaningfully different than the original distribution, in the sense that it is much more likely that $g_V = 0$ in the Poisson-ized version.  This means that the ``natural'' soft-decoding-like condition that one would get out of this is not realizable: The probability that $g_V \in \cout^\perp$ is much bigger than we want it to be, for \emph{any} $\cout$, just because $g_V$ is too likely to be zero.  Fortunately, this seems to be the only obstacle: as in \cref{eq:cond_intro}, we separate out the $g_V = 0$ term (using the analysis from \cref{sec:random}) to arrive at a condition that \emph{is} realizable.  We explain why the condition is realizable---that is, why there exists a $\cout$ that meets it---in \cref{rem:exists}.

\paragraph{\cref{sec:minent}: A smoothness condition on $\cout$.}
For our second sufficient condition, we depart from our moment-based framework and work from first principles.  Our main theorem in \cref{sec:minent} is \cref{main-entropy}, which informally says that if the elements of $\cout$ have ``smooth'' enough distributions of symbols, in the sense that they each have large enough min-entropy, that $\cC = \cout \circ \cin$ will lie near the GV bound with high probability.
The basic idea is to consider a \emph{worst-case} assignment of symbols in $\F_q$ to codewords in $\cin$; this assignment need not be linear and can depend on a particular codeword $c \in \cout$.  Such a worst-case assignment would simply assign the lowest-weight codewords in $\cin$ to the most frequent symbols in a codeword $c \in \cout$.  Using the weight distribution that $\cin$ is likely to have, along with the min-entropy assumption, we can show that this worst-case assignment will \emph{still} result in codewords $w \in \cC$ of weight at least $\frac{1}{2} - O(\eps)$.

We note that, unlike our sufficient condition from \cref{sec:soft}, we don't have a proof of feasibility for our smoothness condition.  That is, as far as we know, there may not be any linear code $\cout$ that is smooth in this sense.  However, as a proof of concept we mention in \cref{rem:exists2} that a random linear code will have a similar property with high probability.   
Moreover, we find it plausible that codewords of \emph{algebraically structured} codes (say, Folded Reed--Solomon codes, Folded Multiplicity, or even large sub-codes of plain Reed--Solomon codes), would satisfy this property, even if a random code does not.

\section{Preliminaries}\label{sec:prelim}

\paragraph{Notation.}  For a vector $x \in \F^n$, and $\alpha \in [n]$, we use $x_\alpha$ to denote that $\alpha$'th entry of $x$.  For $x,y \in \F_2^N$, we define $\langle x,y\rangle = \sum_{\alpha =1}^N x_\alpha y_\alpha \in \F_2$.  For fields $\F_q$ where $q > 2$ is a power of $2$, we use $\langle x,y \rangle = \Tr(\sum_\alpha x_\alpha y_\alpha)$, where $\Tr \colon \F_q \to \F_2$ is the \emph{field trace} defined below.  (In particular it is \emph{not} the standard dot product over $\F_q$!)  The reason for this is explained later, but informally it is because we will think of elements of $\F_q$ as vectors in $\F_2^{k_0}$ when $q = 2^{k_0}$, and our notation matches this.

\paragraph{Codes, Linear Codes, and Random Linear Codes.}
For a finite field $\F$, a code $\cC \subseteq \F^n$ is called \emph{linear} if $\cC$ is a linear subspace of $\F^n$.  The dimension of $\cC$ as a subspace is called the \emph{dimension} of the code.  For a linear code $\cC \subseteq \F^n$, we define the dual code $\cC^\perp \subseteq \F^n$ by
\[ \cC^\perp = \left\{ g \in \F^n \,:\, \sum_{\alpha \in [n]} g_\alpha \cdot c_\alpha = 0 ~ \forall c \in \cC\right\}.\]

We say that a code $\cC\subseteq \F^n$ is a \emph{random linear code} of dimension $k$ if $\cC$ is chosen uniformly among all subspaces of $\F_2^n$ of dimension $k$.

\paragraph{Gilbert--Varshamov Bound.}
In this paper, we study codes that approach the \emph{Gilbert--Varshamov Bound}, which states that there exist codes of distance $\delta$ and rate $R$ approaching $1 - h_2(\delta)$.  
Here, $h_2$ is the binary entropy function, given by
\[ h_2(x) = x \log\left(\frac{1}{x}\right) + (1-x)\log_2\left( \frac{1}{1-x} \right).\]

\begin{theorem}[GV Bound, \cite{Gilbert,Varshamov}]\label{thm:GV}
Let $\delta \in [0,1/2)$ and let $\eta \in (0, 1 - h_2(\delta)]$.
Then for any $n > 1/\eta$, there exists a (linear) code code $\cC \subseteq \F_2^n$ with rate 
\[ R \leq 1 - h_2(\delta) -\eta \]
and relative distance at least $\delta$.
\end{theorem}

We are interested in the parameter regime where the rate $R$ of the code is very small and the distance $\delta$ is very large.  More precisely, we will focus on the setting where $\delta = 1/2 - O(\eps)$.  It is not hard to see (for example, from the Taylor expansion of the entropy function) that in this case
\[ 1 - h_2(1/2 - \eps) = \Theta(\eps^2)\]
as $\eps \to 0$.
Thus, our goal will be the following.
\begin{goal}[What we mean by ``approaching the GV bound'' for low-rate codes]
    In this paper, we say that a family of low-rate codes $\cC_{N,\eps} \subseteq \F_2^N$ ``approaches the GV bound'' if $\cC_{N,\eps}$ has rate $\Omega(\eps^2)$ and distance $1/2 - O(\eps)$, where the asymptotic notation is as $\eps \to 0$ and as $N \to \infty$.  
\end{goal}

\paragraph{Concatenated Codes and Our Default Parameters.}
Throughout the paper, we use $\cout \subseteq \F_q^n$ and $\cin \subseteq \F_2^{n_0}$ as our outer and inner (linear) codes, both of rate $\eps$.  We will let $k = \eps n$ and $k_0 = \eps n_0 = \log(q)$ throughout.  As mentioned in the introduction, we sometimes abuse notation and write $\cout \colon \F_q^k \to \F_q^n$ and $\cin \colon \F_2^{k_0} \to \F_2^{n_0}$ to represent an arbitrary encoding map for these codes.  

Let $N = n_0 \cdot n$ and $K = k_0 \cdot k$.
Abusing notation as noted above,
the concatenated code $\cC = \cout \circ \cin \subseteq \F_2^{N}$ is given by the encoding map $\cC \colon \F_2^{K} \to \F_q^{N}$ defined as follows.  For a message $m \in \F_2^K$, we interpret $m$ as an element of $\F_q^k$ and let $c = \cout(m)$; then $\cC(m)$ is defined as
\[ \cC(m) = (\cin(c_1), \cin(c_2), \cdots, \cin(c_n)) \in \F_2^N.\]

As mentioned in the introduction, we'll take our inner code $\cin \subseteq \F_2^{n_0}$ to be a random linear code of dimension $k_0 = \eps n_0$.  The important property we will need of $\cin$ is that it have about the right weight distribution, which we formalize in the following property.

\begin{definition}[$\tau$-niceness of the inner code]\label{def:nice}  Fix parameters $0 < \tau < \eps$.  We say that the inner code $\cin \subseteq \F_2^{k_0}$ is $\tau$-\emph{nice} if for any $i \in \{1, \ldots, n_0\}$, the number of $c \in \cin^\perp$ so that
$\weight(c) = i$ is at most
\[ \binom{n_0}{i} \cdot 2^{-n_0(\eps - \tau )}. \]
(Notice that we omit the ``$\eps$'' from the name ``$\tau$-nice,'' because $\eps$ will be the same $\eps$ throughout the paper; it is the rate of both $\cin$ and $\cout$).
\end{definition}

\begin{lemma}\label{lem:RLC_nice}
Let $\cin \subseteq \F_2^{n_0}$ be a random linear code of dimension $k_0 = \eps n_0$.  Then with probability at least $1 - n_0\cdot 2^{-\tau n_0}$, $\cin$ is $\tau$-nice.
\end{lemma}
\begin{proof}
    Let $E_i$ denote the event that $\cin^\perp$ has at most ${\binom{n_0}{i}}\cdot 2^{-n_0(\eps - \tau)}$ codewords of weight $i$.  Note that the expected number of such codewords is at most ${\binom{n_0}{i}}2^{-n_0 \eps}$, so by Markov's inequality,
    \[ \Pr_{\cin}[E_i] \geq 1 - 2^{-\tau n_0}.\]
    By a union bound,
    \[ \Pr[\cin \text{ is $\tau$-nice}] \geq 1- n_02^{-\tau n_0},\]
    as claimed.
\end{proof}
We also need the following observation about random binary linear codes.
\begin{lemma}[Negative correlation of $x\in\cC$ and $y\in\cC$]\label{cor:variance}
Let $\cC\subseteq \F_2^n$ be a random linear code and let $x,y\in \F_2^n\setminus \{0\}$ such that $x\ne y$. Then, 
$$\Pr_{\cC}[x,y\in \cC] \le \Pr_{\cC}[x\in \cC]\cdot \Pr_{\cC}[y\in \cC].$$
\end{lemma}
\begin{proof}
    A random linear code of dimension $Rn$ contains a given $d$-dimensional linear subspace of $\F_2^n$ with probability $\prod_{i=0}^{d-1} \frac{2^{Rn}-2^i}{2^n-2^i}$. Since $x$ and $y$ span a $2$-dimensional space, we have    
    $$\Pr[x,y\in \cC] = \frac{\left(2^{Rn}-1\right)\cdot\left(2^{Rn}-2\right)}{\left(2^{n}-1\right)\cdot\left(2^{n}-2\right)} \le \left(\frac{2^{Rn}-1}{2^n-1}\right)^2 = \Pr[x\in \cC]\cdot \Pr[y\in \cC].$$
\end{proof}

\paragraph{$\F_q$ as a vector space over $\F_2$, and Fourier Analysis.}
Let $q = 2^{k_0}$ be a power of $2$; we will set $q$ like this for the rest of the paper.  
It is not hard to see that $\F_q$ is a vector space over $\F_2$ of dimension $k_0$.  In particular, we can identify elements of $\F_2^{k_0}$ with $\F_q$ by simply writing elements of $\F_q$ out in any basis of $\F_q$ over $\F_2$.  For convenience, we will choose a particular basis, which interacts nicely with the \emph{trace map}, defined by
\[ \Tr(x) = x + x^2 + \cdots + x^{2^{k_0 -1}}.\]
It is well-known that $\Tr$ is both $\F_2$-linear and also that its image is indeed $\F_2$.  In that sense, $\Tr(\cdot, \cdot):\F_q \times \F_q \to \F_2$ behaves a bit like a dot product, and in fact this can be made formal.

In more detail, for any $k_0 > 0$ there always exists a \emph{self-dual} basis $\nu_1, \ldots, \nu_{k_0}$ of $\F_q = \F_{2^{k_0}}$ over $\F_2$ (e.g., \cite{SA80,JMV90}); that is, this basis has the property that $\Tr(\nu_i \nu_j) = \ind[i=j]$.  In particular, this means that if we choose such a basis in order to identify $\F_{q}$ with $\F_2^{k_0}$, we have
$$\Tr(\alpha \cdot \beta) = \langle \alpha, \beta \rangle$$ for any $\alpha, \beta \in \F_q \sim \F_2^{k_0}$, where on the left hand side we treat $\alpha$ and $\beta$ as elements of $\F_q$, and on the right hand side we treat them as elements of $\F_2^{k_0}$.\footnote{Indeed, if we write $\alpha = \sum_{i=1}^{k_0} \alpha_i \nu_i$ and $\beta = \sum_{j=1}^{k_0} \beta_j \nu_j$ as elements of $\F_q$, and hence as $(\alpha_1, \ldots, \alpha_{k_0})$ and $(\beta_1, \ldots, \beta_{k_0})$ as elements of $\F_2^{k_0}$, then by $\F_2$-linearity, $\Tr(\alpha \cdot \beta) = \sum_{i,j} \alpha_i \beta_j \Tr(\nu_i \cdot \nu_j) = \sum_i \alpha_i \beta_i = \langle \alpha , \beta \rangle$.}

The reason we want to do our identification between $\F_q$ and $\F_2^{k_0}$ is because it will make the notation a bit easier for \emph{Fourier transforms.}
For a function $\varphi \colon \F^n_q \to \mathbb{R}$, we define the {Fourier transform} of $\varphi$, denoted $\hat{\varphi} \colon \F^n_q \to \mathbb{R}$, by
\[ \hat{\varphi}(\omega) = \frac{1}{q^n} \sum_{x \in \F^n_q} \varphi(x) (-1)^{\Tr\left( \langle \omega , x \rangle \right)}. \]
Notice that if we were to replace the elements of $\F_q$ with their corresponding elements of $\F_2^{k_0}$, and treat $\omega \in \F_q^n$ as $\omega \in \F_2^{k_0 n}$ in the natural way, by the above correspondence we can write this as
\[ \hat{\varphi}(\omega) = \frac{1}{2^{k_0n}} \sum_{x \in \F_2^{k_0n}} \varphi(x) (-1)^{\langle \omega, x \rangle}.\]
Thus, we will drop the trace notation for the rest of the paper, and just use the above definition.  Not only does this simplify the notation for Fourier transforms, but it allows us to move back and forth between elements of $\F_q$ and elements of $\F_2^{k_0}$, which we will want to do anyway as we have to identify elements of $\F_q$ as messages for $\cin$ in $\F_2^{k_0}$.

Next, we record a few useful facts about the Fourier transform.  The first is that the Fourier transform can be inverted: For any function $\varphi \colon \F^n_q \to \mathbb{R}$ and any $x$,
\begin{equation}\label{eq:inverse_fourier_trans}
\varphi(x) = \sum_{\omega \in \F^n_q} \hat{\varphi}(\omega) (-1)^{\langle x, \omega \rangle}.
\end{equation}
The second is that for any $\F_q$-subspace $V \subseteq \F_q^n$, and for any $g \in \F_q^n$, we have
\begin{equation}\label{eq:dualfourier}
\widehat{\ind_V}(g) = \frac{|V|}{q^n} \ind_{V^\perp}(g). 
\end{equation}

\section{A Useful Moment Computation}\label{sec:moment}

In this section, we define a random variable $X_m$ that quantifies how ``bad'' a message $m$ is for our concatenated code, and we prove (\cref{lem:moment}) that the moments of $X_m$ are well-behaved.

We use the notation set up in \cref{sec:prelim}.  Namely, we fix $\eps > 0$, an integer $n$, and a power of two $q = 2^{k_0}$.  We consider the concatenated $\cout \circ \cin$ for linear codes $\cout \subseteq \F_q^n$ of dimension $k = \eps n$ and $\cin \subseteq \F_2^{n_0}$ of dimension $k_0 = \eps n_0$.  We let $N = n \cdot n_0$.

In this section, we think of both $\cout$ and $\cin$ as \emph{fixed}.  The only randomness will be in choosing a random message $m \in \F_q^k \setminus \{0\}$.
Before we begin, we introduce one more definition. 
\begin{definition}\label{def:Omega}
For a code $\cin \subseteq \F_2^{n_0}$ of dimension $k_0$, let $G_0 \in \F_2^{n_0 \times k_0}$ be an (arbitrary) generator matrix for $\cin$.  That is, 
\[ \cin = \{ G_0 \cdot x \,:\, x \in \F_2^{k_0} \}.\]
Let $\Omega \subseteq \F_2^{k_0}$ denote the set of rows of $G_0$, so $|\Omega| = n_0$.  By identifying $\F_2^{k_0}$ with $\F_q$ (as described in \cref{sec:prelim}, using a self-dual basis), we may also treat $\Omega$ as a subset of $\F_q$. 
We say that $\Omega$ is the set \emph{derived from $\cin$.}
\end{definition}
\begin{remark}
    Since $G_0$ can be an arbitrary generator matrix, there is some freedom in defining the set $\Omega$ derived from $\cin$.  It won't matter for the results in this paper, but it may matter for instantiating our sufficient condition in \cref{sec:soft}.
\end{remark}


Next, we define the random variable whose moments we want to bound.
\begin{definition}\label{def:X}
For a message of the outer code $m \in \F_q^k$, define
\[
X_{m} \triangleq \sum_{\alpha \in [n]} \sum_{b \in \Omega}(-1)^{\langle \cout(m)_{\alpha},b \rangle}.
\]
\end{definition}
Observe that $X_m$ is the bias (namely, the difference between the zeros and ones) of the codeword $\code(m) = \cout(m) \circ \cin$. 
In particular, if $X_m = O(\eps N)$ for all nonzero $m$, this will imply that the distance of $\code$ is least $1/2 - O(\eps)$.
Thus, our goal will be to show that $X_m$ is small.
Towards that end,
we will compute the $\moment$-th moment of $X_m$ (over the randomness of a random nonzero message $m$), for a suitable large $\moment$.
Finally, given an ordered list $\mathcal{V}$ of pairs in $[n] \times \Omega$, 
 for every $\alpha \in [n]$ we
denote by $\cV_{\alpha}$ the \emph{multiset}
$\cV_{\alpha} = \set{ b : (\alpha,b) \in \cV }$ (so $b$ appears twice in $\cV_{\alpha}$ if $(\alpha,b)$ appears twice in $\cV$).

The following lemma characterizes the $\moment$-th
moment of $X_m$ in terms of the dual code $\cout^{\perp}$. 
\begin{lemma}\label{lem:moment}
Suppose that $\Omega$ is derived from $\cin$.  Using the notation above, for every $\moment \ge 1$, it holds that
\[
\E_{m \sim \F_q^k \setminus \set{0}}\left[ X_m^{\moment} \right] = \frac{1}{q^{k}-1}\sum_{\cV \in ([n]\times\Omega)^{\moment}}\left( q^k \cdot \ind[g_{\cV}\in \cout^{\perp}] - 1 \right),
\]
where $\cV$ ranges over all tuples $((\alpha_1,b_1),\ldots,(\alpha_{\moment},b_{\moment}))$, and $g_{\cV} \in \F_q^n$ is such that $(g_{\cV})_\alpha = \sum_{b \in \cV_{\alpha}}b$.
\end{lemma}
\begin{remark}[$\F_q$-linear codes over $\F_q^s$.]\label{rem:folding}
    Several constructions of potential outer codes (for example, Folded RS codes~\cite{GR08} or univariate multiplicity codes~\cite{Kop15,GW13}) are not linear over their alphabets but instead are linear over a subfield.  That is, the alphabets for these codes are $\F_q^s$ for some $s > 1$, and the codes are $\F_q$-linear.  An inspection of the proof shows that \cref{lem:moment} still holds in this case: The only thing that changes is that we need to define the dual code $\cout^\perp$ by first ``unfolding'' the original code to treat it as a subspace of $(\F_q)^{sn}$, taking the dual, and ``re-folding.''  
    Intuitively, the proof goes through because the definition of the Fourier transform is the same whether the relevant vectors lie in $(\F_q^s)^n$ or $(\F_q)^{sn}$. 
    Similarly, \cref{lem:moment} holds when $\cout$ is any $\F_2$-linear code (rather than $\F_q$-linear).

    More generally, it is not hard to see that all of the results in the paper go through for $\F_q$-linear codes over $\F_q^s$ (or even any $\F_2$-linear codes) $\cout$.  Indeed, the results in \cref{sec:random} and \cref{sec:soft} essentially rely only on \cref{lem:moment}; while the results in \cref{sec:minent} are separate but are already stated for $\F_2$-linear codes and can easily be seen to extend to larger alphabets.
\end{remark}
\begin{proof}[ of \cref{lem:moment}]
First, we can write, for any $m \in \F_q^k$,
\[
X_m^{\moment} = \sum_{\cV = \langle (\alpha_1,b_1),\ldots,(\alpha_\moment,b_{\moment}) \rangle}~\prod_{j=1}^{\moment}(-1)^{\langle \cout(m)_{\alpha_j},b_j \rangle}
\]
Taking the expectation over a random $m \in \F_q^k \setminus \{0\}$, we have
\begin{align}
\E_{m \sim \F_q^k \setminus \set{0}}\left[ X_m^{\moment} \right] &= \frac{1}{q^k - 1}\sum_{m \in \F_q^k \setminus\set{0}}\sum_{\cV}~\prod_{j=1}^{\moment}(-1)^{\langle \cout(m)_{\alpha_j},b_j \rangle} \nonumber \\
&= \frac{1}{q^k - 1}\sum_{m \in \F_q^k \setminus\set{0}}\sum_{\cV}~\prod_{\alpha \in [n]}(-1)^{\left\langle \cout(m)_{\alpha},\sum_{b \in \cV_{\alpha}}b \right\rangle}, \label{eq:moment1}
\end{align}
where we use the convention that $\sum_{b \in \cV_{\alpha}}b$ is the zero vector whenever $\cV_{\alpha}$ is empty.

Now, for any $\alpha \in [n]$, let \[\varphi_{\alpha}(x) = (-1)^{\left\langle x,\sum_{b \in \cV_{\alpha}}b \right\rangle}.\]
Taking the Fourier transform over $\F_2^{k_0}$, we get for every $w$,
\begin{align*}
\widehat{\varphi_{\alpha}}(w) &= \frac{1}{q}\sum_{x \in \F_2^{k_0}}(-1)^{\left\langle x,\sum_{b \in \cV_{\alpha}}b \right\rangle} \cdot (-1)^{\langle w,x \rangle} \\
&= \frac{1}{q}\sum_{x \in \F_2^{k_0}}(-1)^{\left\langle x,w+\sum_{b\in \cV_{\alpha}}b \right\rangle} = \ind\left[ w = \sum_{b \in \cV_{\alpha}}b \right].
\end{align*}
Let us abbreviate $\overline{m} = \cout(m)$.
Plugging the above back to \cref{eq:moment1}, and using the inverse Fourier transform, we get
\begin{align*}
\E_{m \sim \F_q^k \setminus \set{0}}\left[ X_m^{\moment} \right] &=
\frac{1}{q^k - 1}\sum_{m \in \F_q^k \setminus\set{0}}\sum_{\cV}~\prod_{\alpha \in [n]}\varphi_{\alpha}(\overline{m}_{\alpha}) \\
&= \frac{1}{q^k - 1}\sum_{m \in \F_q^k \setminus\set{0}}\sum_{\cV}~\prod_{\alpha \in [n]}\left( \sum_{w \in \F_2^{k_0}}\widehat{\varphi_{\alpha}}(w) \cdot (-1)^{\langle w,\overline{m}_{\alpha} \rangle}  \right) \\
&= \frac{1}{q^k - 1}\sum_{m \in \F_q^k \setminus\set{0}}\sum_{\cV}\sum_{g \colon [n]\rightarrow \F_{2}^{k_0}}\left( \prod_{\alpha \in [n]} \widehat{\varphi_{\alpha}}(g_{\alpha})  \right)\left( \prod_{\alpha \in [n]} (-1)^{\langle g_{\alpha},\overline{m}_{\alpha} \rangle}  \right).
\end{align*}
Moving the sum over $m$ inside, we have
\begin{align}
\E_{m \sim \F_q^k \setminus \set{0}}\left[ X_m^{\moment} \right] &= \frac{1}{q^k - 1}\sum_{\cV}\sum_{g \in \F_q^n}\left( \prod_{\alpha \in [n]} \widehat{\varphi_{\alpha}}(g_{\alpha})  \right)\left( \sum_{m \in \F_q^k \setminus\set{0}} (-1)^{\sum_{\alpha \in [n]}\langle g_{\alpha},\overline{m}_{\alpha} \rangle}  \right) \notag \\
&=\frac{1}{q^k - 1}\sum_{\cV}\sum_{g \in \F_q^n}\left( \prod_{\alpha \in [n]} \widehat{\varphi_{\alpha}}(g_{\alpha})  \right)\left( \sum_{m \in \F_q^k \setminus\set{0}} (-1)^{\langle g,\overline{m} \rangle}  \right),
\label{eq:moment2}
\end{align}
where in the second equation we have treated $g, \overline{m}$ as elements of $\F_2^{k_0 n}$ in the natural way.
%
Next, we observe that
\begin{equation}\label{eq:momentclaim}
\sum_{m \in \F_q^k \setminus\set{0}} (-1)^{\langle g,\overline{m} \rangle} = \begin{cases}
q^k - 1 & g \in \cout^{\perp}, \\
-1 & \textnormal{otherwise.}
\end{cases}
\end{equation}
Indeed,  we have
\[
\sum_{m \in \F_q^k} (-1)^{\langle g,\overline{m} \rangle} = q^n \cdot \widehat{\ind_{\cout}}(g) = |\cout| \cdot \ind_{\cout^{\perp}}(g) = q^{k} \cdot \ind_{\cout^{\perp}}(g),
\]
and then we subtract off the zero term, 
where the penultimate inequality follows from \cref{eq:dualfourier}.

Thus, given \cref{eq:momentclaim},
 we can write \cref{eq:moment2} as
\begin{equation}\label{eq:moment3}
\E_{m \sim \F_q^k \setminus \set{0}}\left[ X_m^{\moment} \right] = \frac{1}{q^k - 1}\sum_{\cV}\sum_{g \in \F_q^n}\left( \prod_{\alpha \in [n]} \widehat{\varphi_{\alpha}}(g_{\alpha})  \right)\left( q^k \cdot \ind\left[g \in \cout^{\perp}\right] - 1 \right).
\end{equation}
Recalling our expression for $\widehat{\varphi_{\alpha}}$, observe that
\[
\prod_{\alpha \in [n]} \widehat{\varphi_{\alpha}}(g_{\alpha})  = \ind\left[ \forall \alpha \in [n], g_{\alpha} = \sum_{b \in \cV_{\alpha}}b \right].
\]
Thus, in \cref{eq:moment3}, the only nonzero term
in the sum over the $g$-s is $g_{\cV}$, which was indeed defined as $g_{\cV}(\alpha) = \sum_{b \in \cV_{\alpha}}b$. We can then conclude that
\[
\E_{m \sim \F_q^k \setminus \set{0}}\left[ X_m^{\moment} \right] = \frac{1}{q^k - 1}\sum_{\cV}\left( q^k \cdot \ind\left[g_{\cV} \in \cout^{\perp}\right] - 1 \right).
\]
\end{proof}

\section{Most Linear Codes $\cout$ Work Well}\label{sec:random}

In this section we show that there \emph{exist} low-rate concatenated linear codes approaching the GV bound, and in fact most codes $\cout$ will work when concatenated with a random linear inner code $\cin$.  In more detail,
keeping the notation as \cref{sec:prelim} and \cref{sec:moment}, we show
that $\code = \cout \circ \cin$ has distance $\frac{1}{2}-O(\eps)$, with high probability, when both
$\cout$ and $\cin$ are random linear codes of rate $\eps$. Before our proof, we will set further notation that will be useful in later sections as well. 

\begin{definition}\label{def:indicators}
For a sequence of tuples  $\cV = ((\alpha_1,b_1),\ldots,(\alpha_\moment,b_{\moment}))$, 
we define $V = V(\cV) \in \mathbb{N}^{n \times n_0}$ to be the ``unordered'' version of $\cV$, that is,
$V[\alpha,b]$ is the number of times that the pair $(\alpha,b)$ appears in $\cV$. Further, let $B=B(\cV)$ simply be $V ~ \mod~2$, where the modulo is taken element-wise. We refer to the number of $1$-s in the matrix as the \emph{weight} of $B$, denoted $\norm{B}$.
\end{definition}

\begin{theorem}\label{thm:main-random}
There exist constants $c,\bar{c},\tilde{c} > 0$ such that the following holds.
Fix any integer $k > 0$, any $\eps > 0$ sufficiently small (in terms of $\tilde{c}$), and any power-of-two $q = 2^{k_0}$.  Let $n_0 = k_0/\eps$ and $n = k/\eps$, and let $N = n_0 n.$
Suppose that $q \geq 2^{\tilde{c}/\eps^3}$.
%
Let $\cin \subseteq \F_2^{n_0}$ be a linear code of dimension $k_0$ that is $\tau$-nice, for $\tau = 1/\sqrt{n_0}$.  Let $\cout\subseteq \F_q^n$ be an independent random linear code of dimension $k$.
Then, with probability at least $1-
2^{-\eps^2 N/\bar{c}}$ over the choice of $\cout$, the relative distance of $\code = \cout\circ \cin \subseteq \F_2^N$ is at least $\frac{1}{2} - c\cdot \eps$. 
\end{theorem}
\begin{remark}[$\tau$-niceness of inner code]
Note that, by \cref{lem:RLC_nice}, a random linear code $\cin$ is $\tau$-nice for $\tau = 1/\sqrt{n_0}$ with probability at least $1 - 2^{-\Omega(\sqrt{n_0})}$.  Thus, \cref{thm:main-random} implies that $\cout \circ \cin$ approaches the GV bound with high probability over a random $\cin$ and a random $\cout$.
\end{remark}

\begin{proof}
Pick $\moment= \eps^2 N$. 
More precisely, we will choose $\moment$ to be the smallest even integer that is larger than $\eps^2 N$; we assume without loss of generality that $\moment = \eps^2 N$, which will only affect the constants in the theorem statement, and will make the computations much more readable.
\cref{lem:moment} implies that for any fixed $\cin$ (and thus any fixed $\Omega$, the set derived from $\cin$, as in \cref{def:Omega}),
\begin{equation}\label{eq:random1}
\E_{m \sim \F_q^k \setminus \set{0}}\left[ X_m^{\moment} \right] = \frac{1}{q^{k}-1}\sum_{\cV \in ([n]\times\Omega)^{\moment}}\left( q^k \cdot \ind[g_{\cV}\in \cout^{\perp}] - 1 \right).
\end{equation}
In order to make the dependence on $\cout$ and $\cin$ more explicit, we will write $X_m^{\moment}(\cout,\Omega)$.

Our strategy will be to take the expectation of \cref{eq:random1} over the randomness in $\cout$, and use Markov's inequality.
To that end, note that for every sequence of tuples $\cV = \cV(\Omega)$, it holds that
\[
\Pr_{\cout}\left[ g_{\cV} \in \cout^{\perp} \right] =
\begin{cases}
1 & g_{\cV} = 0, \\
\frac{q^{n-k}-1}{q^n-1} & \textnormal{otherwise.}
\end{cases}
\]
This is since
$\cout^{\perp}$ is a random linear code of dimension $n-k$. Therefore, upon taking 
the expectation over $\cout$, for any fixed $\Omega$, \cref{eq:random1} becomes
\begin{align}
\E_{\cout,m \sim \F_q^k \setminus \set{0}}\left[ X_m^{\moment}(\cout,\Omega) \right]  &=
\frac{1}{q^k - 1}\sum_{\cV}\left( q^k \Pr_{\cout}\left[ g_{\cV} \in \cout^{\perp} \right] - 1\right) \nonumber \\
&= \sum_{\cV}\left(\ind\left[ g_{\cV}=0 \right]  - \frac{1}{q^n-1}\cdot \ind\left[g_{\cV}\ne 0\right]\right)\nonumber
\\
&\le \sum_{\cV}\ind\left[ g_{\cV}=0 \right]. \label{eq:random2}
\end{align}
Thus, we are left with bounding the number of $\cV$-s for which $g_{\cV} = 0$, recalling that $(g_{\cV})_\alpha = \sum_{b \in \cV_{\alpha}}b$.

Using the notation of \cref{def:indicators}, each $\cV$ gives rise to $V(\cV)\in \mathbb{N}^{n\times n_0}$ and $B(\cV)\in \F_2^{n\times n_0}$.  Note that $g_{\cV} = 0$ if and only if every row of $B(\cV)$ is a left kernel vector of 
$G_0 \in \F_2^{n_0 \times k_0}$, where $G_0$ is the generating matrix of $\cin$. That is, if and only if every row of $B(\cV)$ belongs to $\cin^\perp$.

We also need the following claim, whose proof we defer. 

\begin{claim}\label{claim:WIsSmall}
Fix a linear code $\cin\subseteq \F_2^{n_0}$ of dimension $k_0$ that is $\tau$-nice for $\tau = 1/\sqrt{n_0}$. Let $\eps \triangleq \frac {k_0}{n_0}$ and let $r \in \mathbb{N}$. 
Let $G_0\in \F_2^{n_0\times k_0}$ be a matrix whose left kernel is $\cin^\perp$, and denote
$$W = \left\{\cV\in ([n]\times \Omega)^r : B(\cV)\cdot G_0 = 0\right\}.$$
Then
$$
|W| \le (8N)^r\cdot (r/N)^{r/2} \cdot 2^{\frac{2N}{\sqrt{n_0}}+ \log_2 N} \cdot \max\left\{ 1, \left(\frac{r e}{\eps^2 N}\right)^{r/2}\right\}.
$$
\end{claim}

Now we finish the proof, given \cref{claim:WIsSmall}. 
Let $\tau = 1/\sqrt{n_0}$ as in the theorem statement.
Let $\cV\in ([n])\times \Omega)^r$ be such that $g_{\cV}=0$. 
Recall that every row of $B(\cV)$ must belong to $\cin^\perp$, so, in particular, $\cV$ must belong to $W$. Hence, for a large enough $n_0$, \cref{eq:random2,claim:WIsSmall} yield 
\begin{align}
\E_{\cout,m \sim \F_q^k \setminus \set{0}}\left[ X_m^{\moment}(\cout,\Omega)  \right] &\le \sum_{\cV}\ind\left[ g_{\cV}=0 \right] \notag\\ 
&= |W| \notag\\
&\le (8N)^r \cdot (re/N)^{r/2} \cdot 2^{2N/\sqrt{n_0} + \log_2 N}\notag\\
&\leq (8N)^r \cdot (e\eps^2)^{r/2} \cdot 2^{2N\eps^2/\sqrt{\tilde{c}}} \cdot N \notag\\
&= \left( 8\sqrt{e} \cdot 2^{2/\sqrt{\tilde{c}}} \cdot  N\eps \right)^r \cdot N \notag \\
&\leq 
\left( 32\sqrt{e} \cdot N\eps \right)^r \cdot N.
\label{eq:grossconstants}
\end{align}
Above, in the third line we have used the fact that $r = \eps^2 N$ to replace the maximum in \cref{claim:WIsSmall} with $e^{r/2}$;
 in the fourth line we have plugged in $r = \eps^2 N$ and also the fact that our assumption $q \geq 2^{\tilde{c}/\eps^3}$ implies that $n_0 \geq \tilde{c}/\eps^4$; and in the last line we have assumed without loss of generality that $\tilde{c} \geq 1$.

By Markov's inequality,
\begin{align*}
&\hspace{-1cm}\Pr_{\cout}\left[\exists m\in \F_q^k\setminus \{0\},~X_m(\cout,\Omega) > c\eps N \right]\\
&\qquad=
\Pr_{\cout}\left[\exists m\in \F_q^k\setminus \{0\},~X_m^{\moment}(\cout,\Omega) > (c\eps N)^r \right] \\
&\qquad\le \sum_{m\in \F_q^k\setminus \{0\}}\Pr_{\cout}\left[X_m^{\moment}(\cout,\Omega) > (c\eps N)^r \right]\\
&\qquad\le \sum_{m\in \F_q^k\setminus \{0\}}\frac{\E_{\cout}\left[X_m^{\moment}(\cout,\Omega) \right]}{(c\eps N)^r} \\
&\qquad=
\frac{\left(q^k-1\right)\cdot \E_{\cout,m \sim \F_q^k \setminus \set{0}}\left[ X_m^{\moment}(\cout,\Omega) \right]}{(c\eps N)^r} \\
&\qquad\le
\frac{2^{r}\cdot \E_{\cout,m \sim \F_q^k \setminus \set{0}}\left[ X_m^{\moment}(\cout,\Omega) \right]}{(c\eps N)^r},
\end{align*}
where in the last line we have used that $q^k = 2^{N\eps^2} = 2^{r}.$
Plugging in \eqref{eq:grossconstants}, we see that this expression is at most
\begin{align*}
    \left( \frac{  64 \sqrt{e} \cdot  N\eps }{c \eps N }\right)^r \cdot N 
    &=  \left( \frac{64 \sqrt{e}}{c} \right)^r \cdot N.
\end{align*}
Thus if we choose $c = 128 \sqrt{e}$, we conclude that
\begin{align*}
\Pr_{\cout}[\exists m \in \F_q^k \setminus\{0\}, X_m(\cout, \Omega) > c \eps N ] &\leq  N \cdot 2^{-r}\\
&\leq  N \cdot 2^{-N\eps^2}.
\end{align*}
As we are assuming that $n_0 \geq \tilde{c}/\eps^4$, we have that
\[ n_0 \leq 2^{n_0 \eps^2}\]
as long as $\eps$ is sufficiently small relative to $\tilde{c}$, and in particular we have
\[ N = n \cdot n_0 \leq 2^{n \cdot n_0 \eps^2} = 2^{N\eps^2}\]
as we always have that $n \geq 1$.
Thus, the above reads
\[ \Pr_{\cout}[\exists m \in \F_q^k \setminus\{0\}, X_m(\cout, \Omega) > c \eps N ] \leq  2^{-N\eps^2 / 2},\]
which proves the claim after observing that without loss of generality we may take $\bar{c} \geq 2$.
\end{proof}

Finally, we prove \cref{claim:WIsSmall}.
\begin{proof}[ of \cref{claim:WIsSmall}]

Let $W' = \left\{V(\cV) : \cV\in W\right\}$. Since $V(\cV)$ preserves all the information in $\cV$ up to ordering, we have $|W|\le r! \cdot |W'|$. We turn to bound $|W'|$. 
Let $V\in W'$. Write $V = B+E$ where $B$ is a $\{0,1\}$ matrix and $E$ is a matrix whose entries are all even. By abuse of notation we also think of $B$ as a matrix in $\F_2^{n\times n_0}$.

Write $p = \frac rN$. Recall that the sum of entries of $V\in W'$ is $Np$. For $0\le m\le Np$, write $W'_m$ for the set of matrices $V=B+E\in W'$ where the weight of $B$ is exactly $m$. We proceed to bound $|W'_m|$ by separately counting the number of ways to choose $E$ and $B$.

\paragraph{Choosing $E$:}    Given a choice of $m$, the matrix $E$ has weight $Np-m$. Each entry of $E$ is even, so each non-zero entry is at least $2$. Thus, there are at most $\frac{Np-m}2$ non-zero entries. The number of ways to choose these entries is thus at most 
    $$\sum_{t=0}^{\frac{Np-m}2} \binom Nt\le 2^{N\cdot h_2(\frac{Np-m}{2N})}.$$ Subject to this choice, the number of possible matrices $E$ is at most 
    $$\binom{\frac{Np-m}2+Np-m-1}{Np-m} \le 2^{\frac{3Np}2}.$$

\paragraph{Choosing $B$:}    Let $I = \{i\in [n] : B_i \ne 0\}$, where $B_i$ stands for the $i$-th row of $B$. The number of ways to choose $I$ is at most $2^n$. Write $|I| = \gamma\cdot n$ for some $\gamma \in [0,1]$. Given $m$ and the choice of $I$, we claim that there are at most \[2^{N\gamma \cdot\left(h_2\left(\frac{m}{\gamma N}\right)-\eps+\frac 1{\sqrt{n_0}}\right)}\] ways to choose the matrix $B$. 
    
    Indeed, recall that every row of $B$ must lie in $\cin^\perp$. Let $B'\in \F_2^{n\times n_0}$ be a random matrix sampled uniformly from all matrices with weight $m$, whose set of non-zero rows is $I$. The number of such matrices is at most $\binom{n_0\cdot |I|}{m} = \binom{\gamma N}{m} \le 2^{\gamma  N\cdot h_2\left(\frac{m}{\gamma N}\right)}$. Conditioning on the weight of each row of $B'$, the fact that $\cin$ is $(1/\sqrt{n_0})$-nice implies that 
    $$\Pr\left[B'_i \in \cin^\perp \mid |B'_i| = w\right] = \frac{|\{x\in \cin^\perp : |x|=w\}|}{\binom {n_0}w}\le 2^{-n_0 \eps + \sqrt{n_0}}$$ for all $i\in I$. Since the rows of $B'_i$ are independent under this conditioning, the probability that every row of $B'$ lies in $\cin^\perp$ is at most $2^{N\gamma \cdot\left(-\eps+\frac 1{\sqrt{n_0}}\right)}$. The desired bound on the number of choices for $B'$ follows.

\paragraph{Obtaining the conclusion:}
    Overall, writing $m = \alpha N$, we conclude that
    \small
    \begin{align*}
        \frac{\log_2(|W'_m|)}N &\le \max_{\gamma \in [0,1]}\left\{\frac{3p}2+\frac nN + h_2\left(\frac{p-\alpha}{2}\right) + \gamma \cdot\left(h_2\left(\frac{\alpha}{\gamma}\right)-\eps+\frac 1{\sqrt{n_0}}\right)\right\} \\
        &\le
        \max_{\gamma \in [0,1]}\left\{\frac{3p}2+\frac nN + \frac{1} {\sqrt{n_0}}+ \left(h_2\left(\frac{p-\alpha}{2}\right) + \gamma \cdot\left(h_2\left(\frac{\alpha}{\gamma}\right)-\eps\right)\right)\right\} \\
        &\le \max_{\gamma \in [0,1]}\left\{\frac{3p}2+\frac 1{\sqrt{n_0}}+\frac nN + \frac{(p-\alpha)}2\log_2\left(\frac{2}{p-\alpha}\right) + (p-\alpha) + \alpha\log_2\left(\frac{\gamma}\alpha\right) + 2\alpha - \gamma \eps\right\} \\
        \\
        &\le \frac{5p}2+\frac 1{\sqrt{n_0}}+\frac nN + \frac{(p-\alpha)}2\log_2\left(\frac{2}{p-\alpha}\right)  + \alpha\left(\log_2\left(\frac{1}\eps\right) + 1-\frac{\ln \ln 2 + 1}{\ln 2}\right)\\
        &\le
        3p+\frac 1{\sqrt{n_0}}+\frac nN + \frac{(p-\alpha)}2\log_2\left(\frac{1}{p-\alpha}\right)  + \alpha\cdot \log_2\left(\frac{1}\eps\right).
    \end{align*}
    \normalsize
     Here, the third inequality is since $h_2(x) \le -x\log_2 x + 2x$ for all $x\in[0,1]$. The fourth inequality is by observing that $\gamma = \frac{\alpha}{\eps\ln 2}$ maximizes the expression.
     
    To bound this last expression, we observe\footnote{
     Indeed, the derivative with respect to $\alpha$ of this expression is
    \[ \frac{1}{2\ln(2)}\left(1 + \ln\left( \frac{p-\alpha}{\eps^2} \right)\right).\]
    We consider two cases, one where $p < \eps^2/e$ and one where $p \geq \eps^2/e$. 
    When $p < \eps^2/e$, 
    \[\ln\left( \frac{p-\alpha}{\eps^2}\right) 
    < \ln(1/e) = -1,\]
    and in particular the derivative is negative for all $\alpha \in [0,p)$.  This means that the maximum is attained at $\alpha = 0$.  
    On the other hand, if $p \geq \eps^2/e$, the derivative is non-negative at zero, and vanishes when $\alpha = p - \eps^2/e$, so the maximum is attained there.
    } that it is maximized at
    \[ \alpha = \max\{p-\eps^2/e, 0\}.\]
   Plugging this in, we claim that
   \[ \frac{\log_2(|W_m'|)}{N} \leq 3p + \frac{1}{\sqrt{n_0}} + \frac{n}{N} + \frac{p}{2}\log_2\left( \max\left\{\frac{1}{p}, \frac{e}{\eps^2}\right\}\right).\]
   In more detail, if $p < \eps^2/e$, then plugging in $\alpha = 0$ yields the $1/p$ term inside the maximum.  On the other hand, if $p \geq \eps^2/e$, then plugging in $\alpha = p - \eps^2 / e$, we have
   \begin{align*}
       \frac{(p-\alpha)}{2} \log_2\left( \frac{1}{p-\alpha} \right) + \alpha \cdot \log_2\left(\frac{1}{\eps}\right) 
       &= \frac{\eps^2 }{2e} \log_2\left( \frac{e}{\eps^2}\right) + \left(p - \frac{\eps^2}{e}\right) \log_2\left(\frac{1}{\eps}\right) \\
       &= \frac{\eps^2}{2 \ln(2) e} + \frac{p}{2} \log_2\left(\frac{1}{\eps^2}\right)\\
       &\leq \frac{p}{2\ln(2)} + \frac{p}{2} \log_2\left(\frac{1}{\eps^2}\right) \qquad \text{as $p \geq \eps^2/e$} \\
       &= \frac{p}{2}\log_2\left( \frac{e}{\eps^2}\right),
   \end{align*}
   which gives us the $e/\eps^2$ term inside the maximum.
    Therefore,
  \begin{align*}
      |W| &\leq (Np)! \sum_{m=0}^{Np}|W_m'| \\
      &\leq N^{Np + 1} \cdot 2^{3Np + N/\sqrt{n_0} + n} \cdot p^{Np/2} \cdot \left(\max\{1, p/\eps^2\}\right)^{Np/2} \\
      &\leq (8N)^r \cdot \left( \frac{r}{N}\right)^{r/2} \cdot 2^{2N/\sqrt{n_0} + \log_2(N)} \cdot \left( \max\left\{1, \frac{r}{\eps^2 N} \right\}\right)^{r/2},
  \end{align*}
    where in the last line we have plugged in $p = r/N$ and also used $n = \frac{N}{n_0} \le \frac{N}{\sqrt{n_0}}$.

\end{proof}

\section{A Soft-Decoding Sufficient Condition for $\cout$}\label{sec:soft}

In this section we give a sufficient condition on $\cout$ for $\cout \circ \cin$ to approach the GV bound.  This condition is similar to a \emph{soft-decoding} condition, except where the distributions for each coordinate are the same.  That is, for a particular distribution $\cD$ on $\F$ (defined in \cref{thm:suff} below), we imagine choosing a random word $x \in \F_q^n$ so that $x_\alpha \sim \cD$ for $\alpha \in [n]$ are all i.i.d.  Then we ask about the probability that $x$ lies in $\cout^\perp$.  If this probability is close to what it ``should'' be (for, say, a random linear code)  then $\cout \circ \cin$ will approach the GV bound with high probability over $\cin$.  In more detail, we prove the following theorem.

\begin{theorem}[Sufficient soft-decoding condition for $\cout$]\label{thm:suff}
There are constants $\tilc, c > 0$ so that the following holds.
Let $\eps > 0$.
  Suppose that $\cin \subseteq \F_2^{n_0}$ is a binary linear code of dimension $k_0 = \eps n_0$ that is $\tau$-nice for $\tau = 1/\sqrt{n_0}$. 
    Further assume that $n_0 \geq 64/\eps^4$.  Let $\Omega \subseteq \F_q$ be the set defined from $\cin$ as per \cref{def:Omega}.

    Let $Y \in \F$ be the random variable given by
    \[ Y = \sum_{b \in \Omega} \zeta_b \cdot b,\]
    where $\zeta_b \sim \ber\left( \frac{1 - e^{-2\tilc \eps^2}}{2}\right)$ are i.i.d.\ Bernoulli random variables, and let $\cD$ be the distribution of $Y$.  Suppose that $\cout \subseteq \F_q^n$ is a linear code of dimension $k = \eps n$ that satisfies
    \begin{equation}\label{eq:softrec} \Pr_{x \sim \cD^n}[ x \in \cout^\perp \setminus \{0\} ] \leq \frac{1}{q^k}\left(1  + \Delta\right),
    \end{equation}
    where 
    \[ \Delta \leq \left( \frac{ c \eps }{2 } \right)^{\tilc \eps^2 N + 100 \sqrt{\tilc \eps^2 N}}.\]
  
    Then, $\cout \circ \cin$ is a binary linear code of rate $\eps^2$ with relative distance at least $\frac{1 - c \eps }{2}$.  
\end{theorem}
\begin{remark}[$\tau$-niceness of inner code]
Note that, by \cref{lem:RLC_nice}, a random linear code $\cin$ is $\tau$-nice for $\tau = 1/\sqrt{n_0}$ with probability at least $1 - 2^{-\Omega(\sqrt{n_0})}$.  Thus, \cref{thm:suff} implies that if $\cout$ satisfies \cref{eq:softrec}, then $\cout \circ \cin$ approaches the GV bound with high probability over a random $\cin$.
\end{remark}

\begin{remark}[The distribution $\cD$ is somewhat concentrated.]\label{rem:listrec}
We note that $(1 - e^{-2\tilc \eps^2})/2 = \Theta(\eps^2)$ for small $\eps > 0$.  In particular, the random variables $\zeta_b$ in the distribution $\cD$ that appears in \cref{thm:suff} is $\ber(p)$ for $p = \Theta(\eps^2)$.
This means that a ``typical''  draw from $\cD$ will be a sum of $\Theta(\eps^2 n_0)$ elements of $\Omega$, so $\cD$ has most of its mass on about $q^{O(\eps \log(1/\eps))}$ elements of $\F_q$, out of $q$.  In this sense, the condition in \cref{thm:suff} is reminiscent of a list-recovery-type condition on $\cout^\perp$, where the input lists are all the same and have size about $\ell = q^{O(\eps \log(1/\eps))}$.
\end{remark}

\begin{remark}[Codes satisfying \cref{eq:softrec} exist]\label{rem:exists}
While the eventual goal is to explicitly construct a code $\cout$ that satisfies \cref{eq:softrec}, as a proof of concept we remark that such codes do exist.
  Indeed, imagine taking a random linear code $\cout$ of dimension $k = \eps n$.  It is not hard to see that
    \[ \E_{\cout} \left[ \Pr_{x \sim \cD^n} [x \in \cout^\perp \setminus \{0\}]\right] = q^{-k} - q^{-n}.\]
    In particular, there exists a linear code $\cout$ of dimension $k$ so that
    \[ \Pr_{x \sim \cD^n}[x \in \cout^\perp \setminus \{0\}] \leq \frac{1}{q^k},\]
    satisfying the condition of \cref{thm:suff} with $\Delta = 0$.
 (Notice that in the theorem, it is okay if  $\Pr_{x \sim \cD^n}[x \in \cout^\perp \setminus \{0\}]$ is smaller than $1/q^k$, it just should not be much larger). 
\end{remark}

We prove the theorem at the end of the section, after putting a few preliminaries in place.  For the rest of the section, let $\tilc$ and $c$ be the constants in \cref{thm:suff}.

For a message $m \in \F^k \setminus \{0\}$, say that $m$ is \emph{bad} if $|X_m| \geq c \eps N$, where $X_m$ is as defined in \cref{sec:moment}.  Thus, if there are no bad messages, then the conclusion of \cref{thm:suff} holds.

For any $r \geq 0$, define
\[ \cB_r(\cout, \cin) \triangleq \frac{ q^k }{q^k - 1} \cdot \frac{1}{(c \eps N)^r} \cdot \sum_{\cV \in (\Lambda \times \Omega)^r} \left(q^k \cdot \ind[g_{\cV} \in \cout^\perp] - 1 \right).\]

\begin{observation}\label{obs:boundbad}
    For any $r$, the number of bad messages $m \in \F^k \setminus \{0\}$ is at most $\cB_r(\cout, \cin)$.
\end{observation}
\begin{proof}
    By \cref{lem:moment}, we have
    \[ \E_{m \sim \F^k \setminus \{0\}}[X_m^r] = \frac{1}{q^k-1} \sum_{\cV \in ([n] \times \Omega)^r} (q^k \cdot \ind[g_{\cV} \in \cout^\perp] - 1). \]
    Markov's inequality along with the definition of $\cB_r(\cout, \cin)$ implies that 
    \[ \Pr_{m \sim \F^k\setminus \{0\}} [X_m \geq C \eps N ] \leq \frac{1}{q^k} \cB_r(\cout, \cin).\]
    Thus, the total number of of bad $m$-s is bounded by $\cB_r(\cout,\cin)$, as desired.
\end{proof}

We may re-write the sum over $\cV \in ([n]\times \Omega)^r$ as an expectation over a corresponding distribution $\cD_r$ on $g \in \F_q^n$.  That is, to sample from $\cD_r$, we choose a sequence $\cV$ of pairs $(\alpha_i, b_i) \in [n] \times \Omega$ for $i \in [r]$ independently and uniformly at random; let $g \in \F_q^n$ be given by
\[ g_\alpha = \sum_{b \in \cV_\alpha} b\]
for $\alpha \in [n]$,
as we did in \cref{sec:moment}.
Thus, the above becomes
\begin{align}
\cB_r(\cout, \cin) &= \frac{ q^k }{q^k - 1} \cdot \left(\frac{N}{c \eps N}\right)^r \cdot \mathbb{E}_{g \sim \cD_r} \left(q^k \cdot \ind[g \in \cout^\perp] - 1 \right) \notag \\
&= \frac{q^k}{q^k - 1} \cdot \left( \frac{1}{c\eps } \right)^r \cdot \left( q^k \Pr_{g \sim \cD_r} [ g \in \cout^\perp] - 1 \right).
\label{eq:badcount1}
\end{align}

This inspires a nice condition for $\cout$; we want $\Pr_{g \sim \cD_r}[ g \in \cout^\perp ]$ to be about $1/q^k$.  This is reminiscent of a soft-decoding problem, but one difference is that  coordinates $g_\alpha$ for $\alpha \in [n]$ are not independent. In order to \emph{make} them independent, as they are in the statement of \cref{thm:suff}, we will first choose $r$ randomly from an Poisson distribution.  That is, 
we will choose
\[ r \sim \poi(\tilc \eps^2 N).\]
We first observe that choosing $r$ at random like this and then choosing $g \sim \cD_r$ results in the product distribution $\cD^n$ in \cref{thm:suff}.
\begin{claim}\label{cl:g_is_prod}
    Suppose that $r \sim \poi(\tilc \eps^2 N)$ and $g \sim \cD_r$.  
    The the joint distribution of the $g_\alpha$ is given by
    \[ g_\alpha  = \sum_{b \in \Omega} \zeta_{\alpha,b} \cdot b,\]
    where
    \[ \zeta_{\alpha,b} \sim \ber\left( \frac{1 - e^{-2\tilc\eps^2}}{2} \right)\]
    are i.i.d.\ Bernoulli random variables.
    In particular, the $g_\alpha$ are all independent and identically distributed for each $\alpha \in [n]$.
\end{claim}
\begin{proof}
We may view the process of choosing the $(\alpha_i, b_i)$ as dropping $r$ balls into $N$ bins, with each bin corresponding to a tuple $(\alpha, b) \in [n] \times \Omega$.  Since $r \sim \poi(\tilc \eps^2 N)$, the occupancy $Y_{\alpha, b}$ of each bin $(\alpha, b)$ is independent, and distributed as 
\[ Y_{\alpha, b} \sim \poi(\tilc \eps^2).\]
Notice that $Y_{\alpha, b}$ is the number of times that $b$ appears in the multiset $\mathcal{V}_\alpha$.  Thus,
 by definition, each $g_\alpha$ is of the form
\[ g_\alpha = \sum_{b \in \Omega} Y_{\alpha, b} b = \sum_{b \in \Omega} \zeta_{\alpha,b} b, \]
where 
\[ \zeta_{\alpha,b} = Y_{\alpha, b} \mathrm{\ mod\ } 2.\]
Above, we can replace $Y_{\alpha, b}$ with $\zeta_{\alpha,b}$ since we are working over a field of characteristic two.

Since the $Y_{\alpha, b}$ are all independent, so are the $\zeta_{\alpha, b}$.
  Further, $\zeta_{\alpha,b}$ is a Bernoulli random variable, and the probability that it is equal to one is the probability that $Y_{\alpha,b}$ is odd.  Since $Y_{\alpha,b} \sim \poi(\tilc \eps^2)$, this is
\[ \Pr[ \zeta_{\alpha,b} = 1 ] = \Pr[ Y_{\alpha, b} \text{ odd}] = \sum_{j=0}^\infty \frac{(\tilc\eps^2)^{2j} e^{-\tilc\eps^2} }{(2j)!} = e^{-\tilc\eps^2}\left( \frac{ e^{\tilc\eps^2} - e^{-\tilc\eps^2}}{2} \right) = \frac{1}{2} \left( 1 - e^{-2 \tilc\eps^2}  \right).\]
This proves the claim.
\end{proof}

Given \cref{cl:g_is_prod}, we return to \eqref{eq:badcount1}, and break it up into two terms, one corresponding to the event that $g = 0$, and one corresponding to the event that $g \neq 0$.  We have
\begin{align}
    \mathcal{B}_r(\cout, \cin) &= \frac{q^k}{q^k - 1} \left( \frac{1}{c \eps}\right)^r \cdot \left( q^k \Pr_{g \sim \cD_r} [g \in \cout^\perp] - 1 \right) \notag\\
    &=\frac{q^k}{q^k - 1} \left( \frac{1}{c \eps}\right)^r \cdot \Pr_{g \sim \cD_r}[g = 0] \cdot \left( q^k \Pr_{g \sim \cD_r}[ g \in \cout^\perp\,|\, g = 0]  - 1 \right)  \notag\\
    & \qquad \qquad +\frac{q^k}{q^k - 1} \left( \frac{1}{c\eps}\right)^r \cdot \Pr_{g \sim \cD_r}[g \neq 0] \cdot \left( q^k \Pr_{g \sim \cD_r}[ g \in \cout^\perp\,|\, g \neq 0 ] - 1 \right)\notag\\
        &= q^k\left( \frac{1}{c \eps}\right)^r \cdot \Pr_{g \sim \cD_r}[g = 0] \notag\\
    & \qquad \qquad +\frac{q^k}{q^k - 1} \left( \frac{1}{c\eps}\right)^r \cdot \Pr_{g \sim \cD_r}[g \neq 0] \cdot \left( q^k \Pr_{g \sim \cD_r}[ g \in \cout^\perp\,|\, g \neq 0 ] - 1 \right)\notag\\
         &= q^k\left(1 + \frac{1}{q^k-1}\right)\left( \frac{1}{c \eps}\right)^r \cdot \Pr_{g \sim \cD_r}[g = 0]  \label{eq:term1}\\
    & \qquad \qquad +\frac{q^k}{q^k - 1} \left( \frac{1}{c\eps}\right)^r  \cdot \left( q^k \Pr_{g \sim \cD_r}[ g \in \cout^\perp \setminus \{0\} ] - 1 \right)\label{eq:term2}
\end{align}
Above, in the second-to-last equality we simplified the first summand by noting that $\Pr_{g \sim \cD_r}[g \in \cout^\perp\,|\, g= 0] = 1$ and canceling the $q^k-1$ terms.
In the final equality, we distributed $\Pr_{g \sim \cD_r}[g \neq 0] = 1 - \Pr_{g \sim \cD_r}[g = 0]$ inside the sum in the second summand, and moved a resulting $\frac{1}{q^k-1} \left( \frac{1}{c\eps}\right)^r \Pr_{g \sim \cD_r}[g = 0]$ term to the first summand.
We handle the two terms \eqref{eq:term1} and \eqref{eq:term2} separately.  The first one we will show is small; and the second we will show is small \emph{if} $\cout$ has a particular soft-decoding-like guarantee.

\begin{claim}\label{cl:first_term_small} 
Suppose that the constants $c,\tilc$ satisfy
\[ \tilc \geq 4 \ln(2) \qquad \text{and} \qquad c \geq 72 \tilc,\]
and that $n$ is sufficiently large.
    Suppose that $r \sim \poi(\tilc \eps^2 N)$ as above, and further suppose that $n_0 \geq 64/\eps^4$.  Then the expectation over $r$ of the term \eqref{eq:term1} satisfies
    \[ \mathbb{E}_r[ \eqref{eq:term1} ] \leq 2^{-\eps^2 N/2}.\]
\end{claim}

\begin{proof}
First, observe that for fixed $r$, the probability that $g \sim \cD_r$ is zero is precisely what is bounded in \cref{claim:WIsSmall},
and we have
\[ \Pr_{g \sim \cD_r}[g = 0] \leq 8^r \cdot \left( \frac{r}{N} \right)^{r/2} \cdot 2^{\frac{2N}{\sqrt{n_0}} + \log_2(N)} \cdot \left( \max\left\{1, \frac{er}{\eps^2 N}\right\}\right)^{r/2}.\]
Let $\lambda = \tilc \eps^2 N$ be the mean of the Poisson distribution that $r$ is drawn from.
Taking the expectation over $r$ and plugging the result of \cref{claim:WIsSmall}, the expected value of \eqref{eq:term1} is:
\begin{align*}
    &\E_r \left[q^k \left( \frac{q^k}{q^k-1}\right)\left( \frac{1}{c\eps} \right)^r \Pr_{g \sim \cD_r} [ g = 0 ] \right]\\
    &\qquad\leq
   q^k 2^{3N/\sqrt{n_0}} \sum_{r \geq 0} \left( \frac{1}{c\eps} \right)^r 8^r \left( \frac{r}{N} \right)^{r/2} \left( \max\left\{ 1, \frac{er}{\eps^2N} \right\}\right)^{r/2} \Pr[\poi(\lambda) = r]\\
   &\qquad=  q^k 2^{3N/\sqrt{n_0}} \sum_{r \geq 0} \left( \frac{1}{c\eps} \right)^r 8^r \left( \frac{r}{N} \right)^{r/2} \left( \max\left\{ 1, \frac{er}{\eps^2N} \right\}\right)^{r/2} \left( \frac{\lambda^r e^{-\lambda}}{r!} \right)\\
   &\qquad=  e^{-\lambda} q^k 2^{3N/\sqrt{n_0}} \sum_{r \geq 0} \frac{1}{r!} \left( \frac{8\lambda \sqrt{r/N}}{c\eps}\right)^r \left( \max\left\{ 1, \frac{er}{\eps^2N} \right\}\right)^{r/2} \\
   &\qquad= e^{-\lambda} q^k 2^{3N/\sqrt{n_0}} \left( \sum_{r=0}^{\eps^2 N/e} \frac{1}{r!} \left( \frac{8\lambda \sqrt{r/N}}{c\eps} \right)^r + \sum_{r>\eps^2 N/e}
   \frac{1}{r!} \left( \frac{8\sqrt{e}\lambda r}{c\eps^2 N} \right)^r.
   \right) 
\end{align*}
Consider each of the two summations above.  The first one is bounded by
\begin{align*}
\sum_{r=0}^{\eps^2 N/e} \frac{1}{r!} \left( \frac{8\lambda \sqrt{r/N}}{c\eps}\right)^r
&\leq \sum_{r=0}^{\infty} \frac{1}{r!} \left( \frac{8\lambda \sqrt{\eps^2/e}}{c\eps}\right)^r \\
&= \exp\left(\frac{8 }{c\sqrt{e}} \cdot \lambda\right).
\end{align*}
Meanwhile, the second term is bounded by 
\begin{align*}
 \sum_{r >  \eps^2 N/e} \frac{1}{r!} \left( \frac{ 8 \sqrt{e} \lambda r}{c\eps^2 N}\right)^r 
 &\leq \sum_{r > \eps^2 N/e} \left( \frac{8 e\sqrt{e} \lambda  r }{c \eps^2 r N }\right)^r \qquad \text{using $r! \geq (r/e)^r$} \\
 &= \sum_{r > \eps^2 N/e} \left( \frac{8e\sqrt{e} \tilc }{c  } \right)^r \\
&\leq 2^{- \eps^2 N/e} < 1
 \end{align*}
 provided that $c  \geq 16e\sqrt{e}\tilc.$  
 Then, we have
\begin{align*}
     \E_r \left[q^k \left( \frac{1}{c\eps} \right)^r \Pr_{g \sim \cD_r} [ g = 0 ] \right] &\leq
    e^{-\lambda} q^k 2^{3N/\sqrt{n_0}} \left( \exp\left( 8 \lambda/c \right) + 1 \right) \\
    &\leq
    2 e^{-\lambda} q^k 2^{3N/\sqrt{n_0}}
    \exp\left( 8\lambda/c \right)  \\
    &\leq \exp\left( \ln(2)\left( 4N/\sqrt{n_0} + \eps^2 N\right) -\lambda\left( 1 -   8/c \right)\right).
\end{align*}
Above, in the final line we used the fact that $q^k = 2^{\eps^2 N}$.  Plugging in $\lambda = \tilc \eps^2 N$, we get
\begin{align*}
     \mathbb{E}_r \left[q^k \left( \frac{1}{c\eps} \right)^r \Pr_{g \sim \cD_r} [ g = 0 ] \right] &\leq
     \exp \left( \eps^2 N \left( \ln(2) - \tilc\left( 1 - 8/c \right) \right) + \frac{4\ln(2) N}{\sqrt{n_0}}  \right).
\end{align*}
Provided that
\[ \tilc \geq 4\ln(2) \qquad \text{and} \qquad c \geq 16,\]
(which both follow from the assumptions in the theorem statement)
we get
\begin{align*}
\mathbb{E}_r \left[q^k \left( \frac{1}{c\eps} \right)^r \Pr_{g \sim \cD_r} [ g = 0 ] \right] &\leq
    \exp_2\left(-\eps^2 N \left( 1 - \frac{4}{\eps^2\sqrt{n_0}} \right)\right) \\
    &\leq \exp_2\left( -\eps^2 N / 2 \right),
\end{align*}
finally using our assumption that $n_0 \geq 64/\eps^4$.  
\end{proof}

Finally we are ready to prove \cref{thm:suff}.
\begin{proof}[\ of \cref{thm:suff}]
The rate of the code $\cout \circ \cin$ follows by definition, so we only need to establish the bound on the relative distance.
By \cref{obs:boundbad}, the number of bad messages is at most $\cB_r(\cout, \cin)$ for any $r$, so it suffices to show that there exists an $r$ so that $\cB_r(\cout, \cin) < 1$.  We will choose $r \sim \poi(\lambda)$, where $\lambda = \tilc \eps^2 N$, as above.  As above, for any $r$ we can write
\[ \cB_r(\cout, \cin) = \eqref{eq:term1} + \eqref{eq:term2},\]
and \cref{cl:first_term_small} implies that
\[ \mathbb{E}_r[ \eqref{eq:term1} ] \leq 2^{-\eps^2 N /2}.\]
In particular, by Markov's inequality, with probability at least $1 - 2^{-\eps^2 N / 4}$ over the choice of $r$, we have
\begin{equation}\label{eq:term1_wellbehaved} \eqref{eq:term1} \leq 2^{-\eps^2 N / 4 }.
\end{equation}
Now we turn our attention to the term \eqref{eq:term2}:
\[ \eqref{eq:term2} = \frac{q^k}{q^k-1} \left( \frac{1}{c\eps}\right)^r \cdot \left( q^k \Pr_{g \sim \cD_r} [ g \in \cout^\perp \setminus \{0\} ] - 1 \right).\]
Note that, by \cref{cl:g_is_prod}, when $r \sim \poi(\lambda)$, the distribution $\cD_r$ is the same as the distribution $\cD^n$ from the statement of the theorem.  Thus,
by the assumption of the theorem, 
\[ q^k \Pr_{r \sim \poi(\lambda), g \sim \cD_r} [ g \in \cout^\perp \setminus \{0\} ] - 1  \leq \Delta \leq \left( \frac{ c\eps }{2} \right)^{\lambda + 100 \sqrt{\lambda}}.\]
Further, by a Chernoff bound for Poisson random variables (e.g., \cite{MU17}),
\[ \Pr[ r \geq \lambda + 100 \sqrt{\lambda} ] \leq \exp\left( \frac{-100^2 \lambda}{\lambda + 100 \sqrt{\lambda} }\right) \leq \exp(-100^2 / 2)\]
for sufficiently large $N$.
Thus, union bounding over the event that \cref{eq:term1_wellbehaved} occurs and that $r \leq \lambda + 100 \sqrt{\lambda}$, we see that with probability at least
\[ 1 - 2^{-\eps^2 N / 2} - \exp(-100^2/2) > 0,\]
over the choice of $r$, we have
\begin{align*}
\cB_r(\cout,\cin) &\leq 2^{-\eps^2 N / 4} + \frac{q^k}{q^k-1} \left( \frac{1}{c\eps}\right)^{\lambda + 100 \sqrt{\lambda}} \cdot \left( \frac{c\eps}{2}\right)^{\lambda + 100 \sqrt{\lambda}} \\
&\leq 2^{-\eps^2 N / 4} + 2^{-\tilc \eps^2 N} \\
&< 1.
\end{align*}
As this is (much) less than $1$ for sufficiently large $N$, we conclude that in particular there exists an $r$ so that $\cB_r(\cout, \cin) < 1$, which proves the theorem.
\end{proof}

\section{A High Min-Entropy Sufficient Condition for $\cout$}\label{sec:minent}


In this section, we give a second sufficient condition under which $\cout$ will be ``good'' for concatenation with a random linear inner code.  This second sufficient condition, informally, says that the codewords of $\cout$ should have ``mildly flat'' symbol distributions.

In more detail, given a word $c \in \F_q^n$, 
let $\cD_c$ denote the empirical distribution of symbols in $c$. 
That is, for $\sigma \in \F_q$, 
\[ \Pr_{\cD_c}[\sigma]=\Pr_{\alpha \sim [n]}[c_\alpha=\sigma]. \]
(For example, if $c = (\sigma, \sigma, \ldots, \sigma)$, then $\bm{c}$ is the distribution on $\F_q$ with 100\% of the mass on $\sigma$; and if $n=q$ and $c$ has one of each different symbol in $\F_q$, then $\bm{c}$ is uniform on $\F_q$).
Given a word $c \in \F_q^n$, the \emph{min-entropy} of the distribution $\cD_c$ is given by 
\[H_{\infty}(c) \triangleq -\log_{2}\max_{\sigma}\Pr_{\cD_c}[\sigma]. \]
Observe that $H_{\infty}(c) \le \log_{2}(q)$. 

Our condition will be about a smoothed notation of min-entropy, which informally allows a small $\eta$-fraction of the mass of $\cD_c$ to have high min-entropy. Formally, for some smoothness parameter $\eta > 0$, we define the \emph{smoothed min-entropy} $H_\infty^\eta$ by
\[
H_{\infty}^{\eta}(\cD_c) \triangleq \max_{\cD_{c'} : \Delta_{\mathrm{TV}}(\cD_c,\cD_{c'})  \le \eta}H_{\infty}(\cD_{c'}),
\]
where for two distributions $\cD_c$, $\cD_{c'}$, $\Delta_{\mathrm{TV}}$ is the total variation distance.

Before we state our main theorem in this section (\cref{main-entropy} below), we state a Lemma that we will eventually apply to the inner code $\cin$.



\begin{lemma}\label{cor:weights-iner}
For any $n \in \mathbb{N}$, $\frac{8}{\sqrt{n}} \le \gamma \le \frac{1}{3}$, and integers $\frac{24\log n}{\gamma} \le k \le \frac{n}{5}$ and $2^{2k/3}\le T \le 2^{(1-\gamma)k}$, a random linear code $\code \subseteq \F_2^n$ of dimension $k$ satisfies the following with probability $1-2^{-\Omega(\gamma k+\gamma^2n+\sqrt{n})}$.

For $0 \le j \le n$, let $\Delta_j$ be the number of codewords of $\code$ of Hamming weight $j$, and denote by $j^{\star}$ the minimal $j$ for which $\sum_{i=0}^{j}\Delta_i \ge T$. Then, 
\begin{enumerate}
    \item\label{it:1} It holds that $j^{\star} \ge \alpha  n$,
 for some $\alpha \ge h_{2}^{-1}\left( 1 - 2 \cdot \frac{k-\log T}{n} \right)$. 
 \item\label{it:2} It holds that 
$\frac{\sum_{i=0}^{j^{\star}}i \cdot \Delta_i}{\sum_{i=0}^{j^{\star}}\Delta_i}  \ge (1-2\gamma)j^{\star}$. 
\item\label{it:3} It holds that $\frac{\Delta_{j^{\star}+1}}{\sum_{i=0}^{j^{\star}}\Delta_i} \le 2^{\sqrt{n}}$. 
\end{enumerate}
\end{lemma}

We defer the proof to the end of the section. Our main result for this section is the following.

 
\begin{theorem}\label{main-entropy}
Fix any sufficiently small $\eps > 0$.
For any integers $k,q \in \mathbb{N}$, so that $q = 2^{k_0}$ is a power of $2$.  Let $n_0 = k_0/\eps$ and $n = k/\eps$.
Let $\cout \subseteq \F_q^n$ be an $\F_2$-linear code of rate $\eps$, and let $\cin \subseteq \F_2^{n_0}$ be a random linear code of rate $\eps$. 

Assume further that there exist constants $\bar{c}_{\gamma},\bar{c}_{\eta}$ such that for every nonzero $c \in \cout$, 
$$H_{\infty}^{\bar{c}_{\eta}\eps}(c) \ge (1-\bar{c}_{\gamma}\eps)\log q,$$
and that  $n_0 = \Omega_{\bar{c}_{\gamma}}\left( \frac{\log(1/\eps)}{\eps^2} \right)$.
Then, with probability at least \[1 - 2^{-\Omega_{\bar{c}_{\gamma}}(\eps^2 n_0 + \sqrt{n_0})}\] over the choice of $\cin$, the concatenated code $\code = \cout \circ \cin$ has relative distance $\frac{1}{2}-O_{\bar{c}_{\gamma},\bar{c}_{\eta}}(\eps)$.
\end{theorem}



\begin{remark}[Do there exist good $\cout$?]\label{rem:exists2}
It is not clear (to us) whether a random linear code satisfies the min-entropy condition of \cref{main-entropy} with high probability (if it did, it would give an alternate proof of \cref{thm:main-random}).
However, as a proof of concept we note that a random linear code does satisfy the property that the symbols of every nonzero codeword are \emph{not} contained in a set of size smaller than $q^{1-\gamma}$. Note that this is necessary from any $c$ with $H_{\infty}(c) \ge (1-\gamma)\log_{2}(q)$.  

Indeed, taking a random linear code $\cout$ of dimension $k = \eps n$, the probability that a given nonzero codeword
violates that constraint is $\sum_{i=1}^{q^{1-\gamma}}\binom{q}{i}(i/q)^{n} \approx q^{-\gamma(n-q^{1-\gamma})}$.
Taking the union bound over all $q^k$ codewords in $\cout$, we get
that no nonzero codewords violates the constraint with high probability,
say for $\gamma = \frac{1}{2}\eps$, assuming $q = O(n)$.
\end{remark}

\begin{proof}[ of \cref{main-entropy}]

Throughout the proof, let 
\[ \gamma = \bar{c}_\gamma \eps  ~\text{  and  }~ \eta = \bar{c}_\eta \eps, \]
where $\bar{c}_\gamma$ and $\bar{c}_\eta$ are the constants from the theorem statement.


As $\cC = \cout \circ \cin$ is a linear code, to lower bound the distance it suffices to lower bound the minimum-weight codeword.  
To that end,
fix any nonzero codeword $c \in \cout$, and consider the word
\[ w = ( \cin(c_1), \cin(c_2), \ldots, \cin(c_n) ).\]


Order the elements of $\F_q$ as
$\sigma_0, \sigma_1, \ldots, \sigma_{q-1}$ so that 
\[ \Pr_{\cD_c}[\sigma_0] \geq\Pr_{\cD_c}[\sigma_1] \geq \cdots \geq \Pr_{\cD_c}[\sigma_{q-1}], \]
and define
\[ p_t(c) = \Pr_{\cD_c}[\sigma_t]\]
for $t \in \{0,1,\ldots, q-1\}$. 

To bound the weight of $w$, it suffices to do so under a worst-case (not necessarily linear) mapping from $\F_q$ to $\cin$.  That is, this worst-case mapping will map the most frequent symbols in $w$ to the lowest-weight codewords in $\cin$. 
Concretely for $j \in \{0,1,\ldots,n_0\}$, let
\[ \Delta_j = \{ x \in \cin\,:\, \weight(x) = j \}\]
be the number of codewords of $\cin$ of weight $j$, and let
\[ \ell_j = \Delta_0 + \Delta_1 + \cdots + \Delta_j\]
be the number of codewords of $\cin$ of weight at most $j$, with the convention that $\ell_{-1} = 0$.
Then we consider a worst-case encoding of $\F_q$ to $\cin$ so that 
the set of symbols
\[ \set{ \sigma_{\ell_{(j-1)}}, \sigma_{\ell_{(j-1)} +1}, \ldots, \sigma_{\ell_{j}-1} }\]
maps to the set of $\Delta_j$ codewords of $\cin$ of weight $j$.
Note that some of these sets may be empty. For example, $\cin$ is not expected have any small-weight nonzero codewords, so, e.g., $\Delta_1$ is very likely zero.

With this notation, we can bound the weight of $w = c \circ \cin$ by $\weight(w) \geq d(c)$, where
\begin{equation}\label{eq:entropy}
d(c)
\triangleq n \cdot \sum_{j=0}^{n_0} \left(\left( \sum_{t=\ell_{j-1}}^{\ell_j - 1} p_t(c) \right) \cdot j\right) =
n \cdot \sum_{j=d_{\mathrm{in}}}^{n_0} \left(\left( \sum_{t=\ell_{j-1}}^{\ell_j - 1} p_t(c) \right) \cdot j\right),
\end{equation}
where $d_{\mathrm{in}}$ denotes the minimum distance of $\cin$.

Next, we lower bound $d(w)$ by considering a worst-case empirical distribution $\cD_c$ for the symbols in $c$.  
That is, we will find the values for $p_t$ that minimize \cref{eq:entropy} while still corresponding to a distribution $\bm{p}$ that obeys our assumption that $H_\infty^{\eta}(\bm{p}) \geq (1 - \gamma)\log_{2}(q)$. 
For convenience, let
\[ b = (1- \gamma) \log q\]
be our bound on the smoothed min-entropy. 

Formally, we have that for all non-zero $c \in \cout$,
\begin{equation}\label{eq:opt_prob}
d(c) \geq d \triangleq \min_{\bm{p} : H_\infty^\eta(\bm{p}) \leq b} n \cdot \sum_{j=d_{\mathrm{in}}}^{n_0} \left( \left( \sum_{t=\ell_{(j-1)}}^{\ell_j - 1} p_t \right) \cdot j \right),
\end{equation}
where the minimum is over all distributions $\bm{p} = (p_0, p_1, \ldots, p_{q-1})$ so that $p_0 \geq p_1 \geq \cdots \geq p_{q-1} \geq 0$ and $H_\infty^\eta(\bm{p}) \leq b$.  

For intuition, consider what this worst-case distribution $\bm{p}$ would be if our requirement were only that the \emph{non-smoothed} entropy $H_\infty(\bm{p}) \geq b$.
Then the worst-case distribution $\bm{p}$ would be the distribution that is uniform on the set of $2^b$ symbols $\sigma$ that map to the lowest-weight codewords; in particular, we'd have $p_t = 2^{-b}$ for $t = 0, \ldots, 2^b -1$.
Given the $\eta$-smoothing allowance, then, the worst-case distribution would simply shift $\eta$ of this mass to the symbol that is mapped to the zero codeword.  This means that the worst-case distribution $\bm{p}$ 
is the one given by values $p_t(b)$ so that:
\[ p_t(b) = \begin{cases} \eta + 2^{-b} & t = 0 \\
2^{-b} & t = 1, \ldots, (1 - \eta)2^b - 1 \\
0 & \text{else} \end{cases}\]

%

For a reason that will be apparent soon, we will in fact work with a slightly smaller entropy bound $b' < b$. Towards defining $b'$, set
\[
T \triangleq (1-\eta)2^{b},
\]
and let
\[ j^{\star} = \min\left\{ j \leq n_0 : \ell_{j^{\star}} \geq T\right\}.\]
We then set $b'$ so that $$\ell_{j^{\star}-1} = (1-\eta)2^{b'} \triangleq T'.$$ Clearly, the worst-case $p_{t}(b')$ is even \emph{worse} than $p_{t}(b)$. Also note that $\frac{T}{T'} \le 1 + \frac{\Delta_{j^{\star}}}{\ell_{j^{\star}-1}}$.

Now that we know what the worst-case $\bm{p}$ looks like, and we still need to bound the value $d$ in \cref{eq:opt_prob}.
Given our expression for $\bm{p}$, we see that 
\begin{equation}\label{eq:d_intermediate} d \geq n \sum_{j=d_{\mathrm{in}}}^{j^{\star} - 1} \left( \left( \sum_{t = \ell_{(j-1)}}^{\ell_j - 1} 2^{-b} \right) \cdot j \right)
=
\frac{n}{2^{b'}} \sum_{j=d_{\mathrm{in}}}^{j^{\star} - 1}(\Delta_j \cdot j)
,\end{equation}


Next, we will apply \cref{cor:weights-iner} with our choice of $\gamma$, and $T'$. 
We conclude that with probability at least
\[ 1 - 2^{-\Omega(\gamma k_0 + \gamma^2 n_0 + \sqrt{n_0})} = 1 -2^{-\Omega_{\bar{c}_\gamma}(\eps^2 n_0 + \sqrt{n_0})} \]
over the choice of $\cin$, the favorable case holds.
(Note that we could indeed apply \cref{cor:weights-iner},
since $\gamma \ge \frac{8}{\sqrt{n_0}}$ and $k_0 \ge \frac{24\log n_0}{\gamma}$ by our lower bound on $n_0$.)

The first conclusion of \cref{cor:weights-iner} implies that $j^{\star} - 1 =\alpha n_0$  for some $\alpha \in [0,1]$ that satisfies
\begin{align*}
\alpha &\geq h_{2}^{-1}\left( 1-2 \cdot \frac{\log q - \log T'}{n_0} \right) = 
h_{2}^{-1}\left( 1-2 \cdot \frac{\log q - \log T + (\log T - \log T')}{n_0} \right) \\
&\ge h_{2}^{-1}\left( 1-2 \cdot \frac{\log q - \log T + \sqrt{n_0}+1}{n_0}\right) \ge 
h_{2}^{-1}\left( 1-2 \cdot \frac{\log q - \log T }{n_0} - \frac{2}{\sqrt{n_0}}\right),
\end{align*}
where the bound on $\log T - \log T'$ follows from the third conclusion of \cref{cor:weights-iner}. Further,
\begin{align*}
\alpha &\geq  h_{2}^{-1}\left(1-2 \eps \cdot \left( \gamma + \frac{\log\left( \frac{1}{1-\eta} \right)}{\log q}\right) - \frac{2}{\sqrt{n_0}} \right) \\
&\ge h_{2}^{-1}\left( 1 - 2\eps\left( \gamma + \frac{2\eta}{ \log q} \right)-\eps \right) \ge  h^{-1}_{2}(1-4\eps\gamma) \ge \frac{1}{2}-\sqrt{\gamma \eps \ln 2}
\end{align*}
where we have used the fact that $h_{2}^{-1}(1-x^2) \le \frac{1}{2}-\frac{\sqrt{\ln 2}}{2}x$ (see, e.g., \cite[Lemma B.2.4]{GRS}), and that $n_0 \ge 4\eps^{-2}$. We also used the fact that we may assume that $\eps \le \frac{1}{2}$ and $q$ is larger than some constant depending on $\bar{c}_\gamma,\bar{c}_{\eta}$. 

Next, the second conclusion of \cref{cor:weights-iner} implies that 
\begin{align*}
 \sum_{j=0}^{j^\star - 1} (\Delta_j ) \cdot j &\geq 
 (1 - 2\gamma) \cdot (j^\star - 1)  \cdot \left(\sum_{j=0}^{j^\star - 1} \Delta_j \right) 
\geq (1 - 2\gamma) \cdot   \alpha n_0 \cdot T',
 \end{align*}
 Thus, returning to \cref{eq:d_intermediate}, we have
  \begin{align*}
 \frac{d}{n_0 \cdot n} &\geq \frac{1}{n_0 \cdot 2^{b'}} \sum_{j=0}^{j^\star - 1} (\Delta_j \cdot j) \\
 &\geq (1 - 2\gamma)\cdot \alpha \cdot \frac{T'}{2^{b'}} = (1-2\gamma) \cdot \alpha \cdot (1-\eta),  
 \end{align*}
 using the definition of $T'$ in the final line.
 Since $N = n_0 n$ is the length of the code $\cC = \cout \circ \cin$, altogether 
 we have that the relative distance of $\cC = \cout \circ \cin$ is at least
 \begin{align*} \frac{d}{N} &\geq (1 - 2\gamma)(1 - \eta)\alpha \\
 &\geq (1 - 2\gamma)(1 - \eta) \left( \frac{1}{2} - \sqrt{ \gamma \eps \ln 2 } \right) \\
 &\geq \frac{1}{2} - \left( \sqrt{\bar{c}_\gamma \ln 2} +  \bar{c}_\gamma + \frac{1}{2} \bar{c}_\eta \right) \eps = \frac{1}{2} - O_{\bar{c}_\gamma, \bar{c}_\eta}(\eps),
 \end{align*}
 as desired.

\end{proof} 

We are left with proving \cref{cor:weights-iner}.

\begin{proof}[ of \cref{cor:weights-iner}]
Similarly to what we did in the proof of \cref{thm:main-random}, we will work with codes whose weight distribution deviates only slightly from that of a random code. While \cref{thm:main-random} required only an upper bound on the weights, here we also need a lower bound. 

Write $p = 2^{-(n-k)}$ and $p' =2^{-(n-k)}\cdot \left(1-2^{-k}\right)$ and fix $1\le i\le n$. Observe that 
$$\binom ni \cdot p' \le \E_{\cC}[\Delta_i] = \binom ni\cdot \frac{2^{k}-1}{2^n-1} \le \binom ni \cdot p$$
and 
$$\Var_{\cC}[\Delta_i] \le \binom ni \cdot p,$$
where the latter follows from  \cref{cor:variance}.
Fix a $\tau > 0$ to be determined soon.
Markov's inequality now yields 
$$
\Pr_{\cC}\left[\Delta_i \ge 2^{\tau n} \cdot \binom ni\cdot p\right] \le  {2^{-\tau n}}
$$
and Chebyshev's inequality yields
$$
\Pr\left[\Delta_i \le \frac{\binom ni\cdot p}{2^{\tau n}} \right] \le \frac{1}{(1-2^{-\tau n}-2^{-k})^2\cdot \binom ni\cdot p}.
$$
Fixing any $1\le \ell\le \frac n2$ and taking a union bound, it holds with probability at least $1 - \frac n{2^{\tau n}} - \frac{n}{(1-2^{-\tau n}-2^{-k})^2\cdot \binom n{\ell}\cdot p}$ that \begin{equation}\label{eq:weights-inerUpperBound}
    \Delta_i \le 2^{\tau n}\cdot \binom ni\cdot p \quad \text{ for all }1\le i\le n
\end{equation} 
and 
\begin{equation}\label{eq:weights-inerLowerBound}
    \Delta_i \ge \frac{\binom ni\cdot p}{2^{\tau n}}\quad \text{for all } \ell\le i\le \frac n2.
\end{equation}
Note that \cref{eq:weights-inerUpperBound} means that $\code^{\perp}$ is $\tau$-nice, in the sense of \cref{def:nice}.
Take 
\[
\tau = \min\set{\frac{1}{8}\gamma^2\cdot \left(h_2^{-1}\left(1-\frac kn\right)\right)^2,\frac{k\cdot \gamma}{n} -\frac{\log n}n},~\ell = \left\lfloor{h_2^{-1}\left(1-\frac 23\cdot \frac{k}{n}\right)n}\right\rfloor.\] 

We bound the probability that \cref{eq:weights-inerLowerBound,eq:weights-inerUpperBound} then hold simultaneously, according to which of the two terms above minimize $\tau$. If it is the first one, 
then the probability is at least
$$
1 - \frac{n}{2^{\frac{1}{8}\gamma^2\cdot \left(h_2^{-1}\left(1-\frac kn\right)\right)^2\cdot n}} - \frac{n}{\left(1-2^{-\frac{1}{8}\gamma^2\cdot \left(h_2^{-1}\left(1-\frac kn\right)\right)^2\cdot n}-2^{-k}\right)^2\cdot \binom n\ell\cdot p} \triangleq 1-\delta_1-\delta_2.
$$
To bound $\delta_1$, we use the fact that $h_{2}^{-1}(1-x) \ge \frac{1}{2}-5x^2$ whenever $x \le \frac{1}{4}$  \cite[Lemma B.2.4]{GRS}, and get $n\delta_1 \le 2^{-(\gamma^2/128)n}$, since $h_{2}^{-1}(1-\frac{k}{n}) \ge \frac{1}{4}$ follows from our upper bound on $k$.
To bound $\delta_2$, first note 
that $1-2^{-\gamma^2\cdot \left(h_2^{-1}\left(1-\frac kn\right)\right)^2\cdot n}-2^{-k} \ge \frac{1}{2}$ since $\gamma$ is large enough.
Next, $\binom{n}{\ell} \ge \frac{1}{\sqrt{2n}}2^{(1-\frac{2k}{3n})n}$,
so $\binom{n}{\ell}p \ge \frac{1}{\sqrt{2n}}2^{k/3} \ge 2^{k/4}$ (using our lower bound on $k$). We then have $\delta_2 \le n \cdot 2^{-k/4}$,
and overall, $\delta_1 + \delta_2 \le 2^{-(\gamma^{2}/256)n}+2^{-k/6}$, again using our lower bound on $k$ and the fact that $\gamma$ is large enough.


Next, we bound the probability that \cref{eq:weights-inerLowerBound,eq:weights-inerUpperBound} in the case where $\tau = \frac{k \cdot \gamma}{n} - \frac{\log n}{n}$, namely,
\[
1 - \frac{n}{2^{\left( \frac{k \cdot \gamma}{n} - \frac{\log n}{n} \right)n}}
- \frac{n}{\left( 1 - 2^{-\left( \frac{k \cdot \gamma}{n} - \frac{\log n}{n} \right)n} - 2^{-k} \right)^{2} \cdot \binom{n}{\ell} \cdot p} \triangleq 1 - \delta'_1 - \delta'_2
\]
In this case, $\delta'_1$ can be upper bounded 
by $2^{-(\gamma/4)k}$ since $\frac{\log n}{n} \le \frac{k \cdot \gamma}{n}$
and by our lower bound on $k$. For $\delta'_2$, we again use the
fact that $\binom{n}{\ell}p \ge 2^{k/4}$, and we also  have that
$1-2^{-\tau\cdot n}-2^{-k} \ge \frac{1}{2}$, again from our lower bound on $k$. Thus, $\delta'_1 + \delta'_2 \le 2^{-\Omega(\gamma k)}$.

Assuming \cref{eq:weights-inerLowerBound,eq:weights-inerUpperBound} hold, we show that they imply \cref{it:1,it:2}, and begin with \cref{it:1}.
By \cref{eq:weights-inerUpperBound}, in order to bound $j^\star$, it suffices to solve the following equation for $j$:
\[
\sum_{i=0}^{j}\binom{n}{i}2^{-(n-k-\tau n)} \ge T.
\]
Writing $j = \alpha n$  and using standard bounds on the sum of binomial coefficients, we need to find the smallest $\alpha$ for which $\alpha n$ is an integer, and
\[
2^{h_{2}(\alpha)n - \frac{1}{2}\log n - 1} \cdot 2^{-n+k+\tau n} \ge T.
\]
Thus, the $\alpha$ for which $j^{\star} = \alpha n$ satisfies
\[
\alpha \ge h_{2}^{-1}\left( 1 - \tau - \frac{k-\log T+\log n}{n}  \right)\ge h_{2}^{-1}\left( 1 - 2\cdot \frac{k-\log T}{n}  \right),
\]
since $\tau \le \frac{k\cdot \gamma}{n} - \frac{\log n}n \le \frac{k-\log T - \log n}n.$

We now prove \cref{it:2}. Let $j' = \lceil\alpha\cdot(1-\gamma)\cdot n\rceil$. By \cref{eq:weights-inerUpperBound},
$$A \triangleq \sum_{i=0}^{j'-1}\Delta_i \le \sum_{i=0}^{j'-1} \binom ni\cdot 2^{-n+k+\tau n} \le 2^{n\left(h_2\left(\alpha(1-\gamma)\right) + \tau + \frac kn - 1\right)}.$$
By \cref{it:1} and our assumption that $T \ge 2^{\frac {2k}3}$, 
$$j^{\star} = \alpha \cdot n  \ge h_2^{-1}\left(1-2\cdot \frac{k-\log T}n\right)\cdot n \ge h_2^{-1}\left(1-\frac 23\cdot \frac kn\right)\cdot n\ge \ell.$$
Hence, \cref{eq:weights-inerLowerBound} yields
$\Delta_{j^{\star}} \ge \binom ni\cdot p\cdot 2^{-\tau n}$, and so, 
$$B \triangleq \sum_{i=j'+1}^{j^*} \Delta_i \ge \Delta_{j^{\star}} \ge \binom n{j^{\star}}\cdot p\cdot 2^{-\tau n}\ge 2^{n\cdot\left(h_2(\alpha)-\frac {\log(2n)}{2n}-1+\frac kn-\tau\right)}.$$
Therefore,
$$\frac{\sum_{i=0}^{j^{\star}} i\cdot \Delta_i}{\sum_{i=0}^{j^{\star}} \Delta_i} \ge \frac{j'\cdot \sum_{i=j'}^{j^{\star}}\Delta_i}{\sum_{i=0}^{j^{\star}}\Delta_i} = j'\cdot\frac{B}{A+B} \ge (1-\gamma)\cdot j^{\star}\cdot \frac{1}{1+\frac AB}.$$
Now,
\begin{align*}
\frac {\log \left(\frac AB\right)}n &\le h_2(\alpha(1-\gamma)) - h_2(\alpha) + 2\tau + \frac{\log(2n)}{2n} \\  &\le h_2(\alpha(1-\gamma)) - h_2(\alpha) + 4\tau &\text{by the lower bound on $\gamma$} \\ &\le -(\alpha \gamma)^2 + 4\tau &\text{since $h_{2}'(x) \ge 0$ and $h_{2}''(x) \le -2$}\\
&\le -(\alpha \gamma)^2 + \frac{1}{2} \gamma^2\cdot  \left(h_2^{-1}\left(1-\frac kn\right)\right)^2
\\ &\le  -\frac{1}{2} \gamma^2 \cdot \left(h_2^{-1}\left(1-\frac kn\right)\right)^2 \triangleq -\theta. &\text{by \cref{it:1} and the assumption }T\ge 2^{\frac {2k}3}.
\end{align*}
We conclude that
$$ \frac{\sum_{i=0}^{j^{\star}} i\cdot \Delta_i}{\sum_{i=0}^{j^{\star}} \Delta_i} \ge \frac{(1-\gamma)j^{\star}}{1+2^{-\theta n}}.$$
All that is left is to show that $2^{-\theta n} \le \gamma$. Indeed,
this easily follows from our lower bound on $k$.

Finally, let us prove \cref{it:3}. Following the same reasoning as above, for some $\tau > 0$, we have that
$
\Delta_{j^{\star}+1} \le 2^{\tau n}\cdot \binom{n}{j^{\star}+1}\cdot p
$
and
$
\Delta_{i} \ge 2^{-\tau n}\cdot \binom{n}{j^{\star}}\cdot p
$
with probability at least 
\begin{equation}\label{prob:it3}
1 - \frac {1}{2^{\tau n}} - \frac{1}{(1-2^{-\tau n}-2^{-k})^2\cdot \binom {n}{j^{\star}}\cdot p},
\end{equation} where $p = 2^{-n+k}$.
Now, we readily get
\[
\frac{\Delta_{j^{\star}+1}}{\sum_{i=0}^{j^{\star}+1}\Delta_i} \le 
\frac{\Delta_{j^{\star}+1}}{\Delta_{j^{\star}}} \le 
\frac{2^{\tau n}\cdot \binom{n}{j^{\star}+1}\cdot p}{2^{-\tau n}\cdot \binom{n}{j^{\star}}\cdot p} \le 2^{2\tau n} \cdot \frac{n-j^{\star}}{j^{\star}+1} \le 2^{2\tau n + 2},
\]
using the fact that $j^{\star} \ge \frac{1}{4}n$.
Set $\tau = \frac{1}{4\sqrt{n}}$. The above bound is thus at most $2^{\sqrt{n}}$, and
the success probability, \cref{prob:it3}, is at least
\[
1-2^{-\frac{1}{4}\sqrt{n}} - \frac{1}{4} \cdot 2^{-k/4},
\]
where we used $1-2^{-\tau n} - 2^{-k} \ge \frac{1}{2}$, and 
\[
\binom{n}{j^{\star}}2^{-n+k} \ge 2^{\left( h_{2}(\alpha) -1 + \frac{k}{n} - \frac{2\log n}{n} \right)n} \ge 2^{\left(  \frac{k}{n} - 2 \cdot \frac{k-\log T}{n} - \frac{2\log n}{n} \right)n} \ge 
2^{\left(  \frac{k}{3n} - \frac{2\log n}{n} \right)n} \ge 2^{-k/4}.
\]




\end{proof}

\subsection*{Acknowledgements}
We thank Amnon Ta-Shma for helpful and interesting discussions, and
collaboration at the beginning of this work. 
We thank Arya Mazumdar for pointing out~\cite{BJT01} and for helping us understand its implications.
This work was done partly while the authors were visiting the Simons Institute for the Theory of Computing.

\bibliographystyle{alpha}
\bibliography{refs}

\end{document}